\documentclass{article}

\usepackage[preprint]{neurips_2025}

\usepackage[utf8]{inputenc} % allow utf-8 input
\usepackage[T1]{fontenc}    % use 8-bit T1 fonts
\usepackage{hyperref}       % hyperlinks
\usepackage{url}            % simple URL typesetting
\usepackage{booktabs}       % professional-quality tables
\usepackage{amsfonts}       % blackboard math symbols
\usepackage{nicefrac}       % compact symbols for 1/2, etc.
\usepackage{microtype}      % microtypography
\usepackage{xcolor}         % colors
\usepackage[english]{babel}
\usepackage[utf8]{inputenc}
\usepackage{acro}
\usepackage{adjustbox}
\usepackage[linesnumbered,ruled,lined,noend]{algorithm2e}
\SetKwComment{Hline}{}{\vspace{-3mm}\textcolor{gray}{\hrule}\vspace{1mm}}
\SetKwProg{Fn}{function}{}{}
\usepackage{algorithmic}

\makeatletter
\let\cref@old@stepcounter\stepcounter
\def\stepcounter#1{%
  \cref@old@stepcounter{#1}%
  \cref@constructprefix{#1}{\cref@result}%
  \@ifundefined{cref@#1@alias}%
    {\def\@tempa{#1}}%
    {\def\@tempa{\csname cref@#1@alias\endcsname}}%
  \protected@edef\cref@currentlabel{%
    [\@tempa][\arabic{#1}][\cref@result]%
    \csname p@#1\endcsname\csname the#1\endcsname}}
\makeatother

\usepackage{amsfonts}
\usepackage{amsmath}
\usepackage{amsthm}
\usepackage{blindtext}
\usepackage{bm}
\usepackage{centernot}
\usepackage{changepage}
\usepackage{csquotes}
\usepackage{float}
\usepackage{lipsum}
\usepackage{mathtools}
\usepackage{makecell}
\usepackage{multicol}
\usepackage{ragged2e}
\usepackage{rotating}
\usepackage{scalerel}
\usepackage{stackengine}
\usepackage{subfig}
\usepackage{tablefootnote}
\usepackage{tabularx}
\usepackage{thm-restate}
\usepackage{thmtools}
\usepackage{wasysym}
\usepackage{xargs}
\usepackage{xspace}
\usepackage{xparse}
\usepackage{xpatch}
\usepackage{url}            % simple URL typesetting
\usepackage{nicefrac}       % compact symbols for 1/2, etc.
\usepackage{todonotes}
\usepackage{paracol}
\usepackage{amssymb}
\usepackage{siunitx}  % For table alignment
\usepackage[most]{tcolorbox}
\usepackage{framed}

\graphicspath{
  {./},
  {./images/},
  {../images/},
  {../../images},
}

\newtheorem{customtheorem}{Theorem}

\definecolor{darkgreen}{RGB}{47,182,56}

\providecommand{\corollaryname}{Corollary}
\providecommand{\claimname}{Claim}
\providecommand{\definitionname}{Definition}
\providecommand{\lemmaname}{Lemma}
\providecommand{\notationname}{Notation}
\providecommand{\remarkname}{Remark}
\providecommand{\problemname}{Problem}
\providecommand{\propositionname}{Proposition}
\providecommand{\examplename}{Example}
\providecommand{\theoremname}{Theorem}
\providecommand{\conjecturename}{Conjecture}
\providecommand{\experimentname}{Experiment}

\renewcommand{\cite}[1]{{\color{red}USE CITEP}}

\DeclareMathAlphabet{\mathpzc}{OT1}{pzc}{m}{it}
\DeclareMathSymbol{\shortminus}{\mathbin}{AMSa}{"39}

\newcommand{\R}{\mathbb{R}}
\newcommand{\mc}{\mathcal}

\newcommand{\twoPlLeduc}{{\tt biased\_2p\_leduc}}
\newcommand{\Rho}{\mathrm{P}}

\newcommand{\vv}[1]{%
 \boldsymbol{#1}
}

\newcommand{\defeq}{\ensuremath{\stackrel{\rm def}{=}}}

\newcommand{\inftynorm}[1]{\left\lVert#1\right\rVert_\infty}  % Norm ||x||
\newcommand{\inner}[2]{\langle #1, #2 \rangle}    % Inner product
\newcommand{\opp}{{\textnormal{-} i}}             % opponent
\newcommand{\br}{{\mathpzc{br}}}                  % best response
\newcommand{\defword}[1]{\textbf{\boldmath{#1}}}

\newcommand{\priHist}{\ensuremath{s_\pl}}

\newcommand{\priTree}{\ensuremath{\mc S_\pl}}
\newcommand{\privTree}{\ensuremath{\mc S_\pl}}

\newcommand{\bandit}{\ensuremath{m}}
\newcommand{\task}{\ensuremath{g}}
\newcommand{\tasks}{\ensuremath{G}}
\newcommand{\predictor}{\ensuremath{\pi}}
\newcommand{\prediction}{\ensuremath{\vv{p}}}

\newcommand{\reward}{\ensuremath{\vv{x}}}
\newcommand{\regret}{\ensuremath{\vv{r}}}
\newcommand{\cumregret}{\ensuremath{\vv{R}}}
\newcommand{\predregret}{\ensuremath{\vv{\xi}}}

\newcommand{\actionset}{\mathcal{A}}
\newcommand{\CFR}{{\sc CFR}}
\newcommand{\PCFR}{{\sc PCFR}}
\newcommand{\SPCFR}{{\sc SPCFR}}

\newcommand{\NPCFR}{{\sc NPCFR}}

\DeclareMathOperator{\EX}{\mathbb{E}}% expected value
\DeclareMathOperator*{\argmax}{arg\,max}

\definecolor{insignwin}{rgb}{0.5,0.5,0.5}
\definecolor{insignlose}{rgb}{0.5,0.5,0.5}
\definecolor{win}{rgb}{0,0,0}
\definecolor{lose}{rgb}{0,0,0}

\newcommand{\strategy}{\ensuremath{\bm{\sigma}}}
\newcommand{\mstrategy}{\ensuremath{\bm{s}}}

\newcommand{\pstrategy}{\ensuremath{\bm{\rho}}}
\newcommand{\pstrategySet}{\ensuremath{\bm{\Rho}}}
\newcommand{\jointstrategy}{\ensuremath{\bm{\delta}}}
\newcommand{\strategySet}{\ensuremath{\bm{\Sigma}}}

\newcommand{\pl}{i}

\newcommand{\worldstate}{\ensuremath{\mc W}}
\newcommand{\jointaction}{\ensuremath{\mc A}}

\newtheorem{theorem}{\protect\theoremname}
\newtheorem{definition}[theorem]{\protect\definitionname}
\newtheorem{lemma}[theorem]{\protect\lemmaname}

\title{Approximating Nash Equilibria in General-Sum Games via Meta-Learning}

\author{%
  David Sychrovsk\'{y} \\
  Charles University\\
  Prague, Czechia \\
  \texttt{sychrovsky@kam.mff.cuni.cz} \\
  \And
  Christopher Solinas \\
  University of Alberta\\
  Edmonton, Canada \\
  \texttt{solinas@ualberta.ca} \\
  \And
  Revan MacQueen \\
  Alberta Machine Intelligence Institute\\
  Edmonton, Canada \\
  \texttt{revan.macqueen@amii.ca} \\
  \And
  Kevin Wang \\
  Brown University \\
  Rhode Island, United States \\
  \texttt{kevin\_a\_wang@brown.edu} \\
  \AND
  James R.\ Wright \\
  University of Alberta\\
  Edmonton, Canada \\
  \texttt{james.wright@ualberta.ca} \\
  \And
  Nathan R. Sturtevant \\
  University of Alberta\\
  Edmonton, Canada \\
  \texttt{nathanst@ualberta.ca} \\
  \And
  Michael Bowling \\
  University of Alberta\\
  Edmonton, Canada \\
  \texttt{mbowling@ualberta.ca} \\
}

\begin{document}

\maketitle

\begin{abstract}
% Nash equilibrium is perhaps the best-known solution concept in game theory. 
% Such a solution assigns a strategy to each player which offers no incentive to unilaterally deviate.
% % If the player are in equilibrium, no incentive to deviate.
% While one is guaranteed to always exist, the problem of finding one in general-sum games is PPAD-complete, generally considered intractable.
% Regret minimization is an efficient framework for approximating equilibria.
% However, such algorithms are only guaranteed to converge to a coarse-correlated equilibrium (CCE).
% % Despite this negative worse-case result, algorithms for finding Nash equilibria empirically work well.
% The resulting strategy profile can correlate actions among players, making it impossible to reason about how each player should act independently.
% In this work, we use meta-learning to minimize the correlations in the strategy profile produced by a regret minimizer.
% The meta-learned regret minimizer is still guaranteed to converge to a CCE, but we give a bound on the distance to Nash equilibrium in terms of our meta-loss.
% % and a low meta-loss provably leads to a closer approximation of a Nash equilibrium. %, and is meta-learned to converge closer to a Nash equilibrium.
% % We show our approach is sound by upper-bounding the distance to a Nash equilibrium by the meta-loss.
% We evaluate our approach in general-sum imperfect information games.
% Our algorithms provide significantly better approximations of Nash equilibria than state-of-the-art regret minimization techniques.

% V2

Nash equilibrium is perhaps the best-known solution concept in game theory. 
Such a solution assigns a strategy to each player which offers no incentive to unilaterally deviate.
While a Nash equilibrium is guaranteed to always exist, the problem of finding one in general-sum games is PPAD-complete, generally considered intractable.
Regret minimization is an efficient framework for approximating Nash equilibria in two-player zero-sum games.
However, in general-sum games, such algorithms are only guaranteed to converge to a coarse-correlated equilibrium (CCE), a solution concept where players can correlate their strategies.
In this work, we use meta-learning to minimize the correlations in strategies produced by a regret minimizer. This encourages the regret minimizer to find strategies that are closer to a Nash equilibrium.
The meta-learned regret minimizer is still guaranteed to converge to a CCE, but we give a bound on the distance to Nash equilibrium in terms of our meta-loss.
We evaluate our approach in general-sum imperfect information games.
Our algorithms provide significantly better approximations of Nash equilibria than state-of-the-art regret minimization techniques.
\end{abstract}

\section{Introduction}
\label{sec: intro}

The Nash equilibrium is one of the most influential solution concepts in game theory. 
%It assigns a strategy for each player with the guarantee that no player can benefit by unilaterally deviating from it.
A strategy profile is a Nash equilibrium if it has the guarantee that no player can benefit by unilaterally deviating from it.
The robustness of this guarantee means that Nash equilibria have applications in many domains ranging from economics~\citep{vickrey1961counterspeculation,milgrom1982theory} to machine learning~\citep{goodfellow2014generative}.
Finding an efficient algorithm for computing Nash equilibria has attracted much attention~\citep{rosenthal1973class,monderer1996potential,kearns2001graphical,cai2011minimax,munoz2012polynomial}. 
However, it was shown that, in its full generality, finding a Nash equilibrium is PPAD-complete~\citep{papadimitriou1994complexity,daskalakis2009complexity}.
Many related decision problems, such as `Is a given action in the support of a Nash equilibrium?', are NP-complete~\citep{gilboa1989nash}.

Despite these negative results, computing Nash equilibria in special classes of games, in particular two-player zero-sum games, is tractable.
In this setting, regret minimization has become the dominant approach for finding Nash equilibria~\citep{nisan2007algorithmic}.
This framework casts each player as an independent online learner who repeatedly interacts with the game, selecting strategies according to dynamics that lead to sublinear growth of their accumulated \textit{regret}.
Regret minimizers guarantee convergence to Nash equilibria in two-player zero-sum games, and are the basis for many significant results
in imperfect information games~\citep{bowling2015heads,DeepStack,brown2018superhuman,brown2020combining,Pluribus,schmid2023student}.

Outside the two-player zero-sum setting, regret minimization algorithms are no longer guaranteed to converge to a Nash equilibrium.
Instead, a regret minimizer's empirical distribution of play converges to a coarse-correlated equilibrium (CCE)~\citep{hannan1957approximation, hart2000simple}. 
The CCE is a relaxed equilibrium concept, which gives a distribution over the \emph{outcomes} of the game such that it isn't beneficial for any player to deviate from it. 
If this distribution is uncorrelated, meaning it can be expressed as a profile of independent strategies, it is also a Nash equilibrium.
As such, Nash equilibria form a subset of CCEs, for which the outcome distribution can be marginalized into strategies of the individual players. 
The degree to which a CCE is correlated, or how much a player can infer about the actions of other players given their action, can be formalized by total correlation~\citep{watanabe1960information}. 

% Recently, \citep{farina2021faster} introduced the predictive regret minimization algorithm, which is guaranteed to minimize regret for arbitrary bounded predictions.
A recently proposed \emph{learning not to regret} framework allows one to meta-learn a regret minimizer to optimize a specified objective, while keeping regret minimization guarantees~\citep{sychrovsky2024learning}.
Their goal was to accelerate the empirical convergence rate on a distribution of black-box tasks.
%In this work, we rather minimize the correlation in the players' strategies.
In this work, we meta-learn predictions that optimize an alternative meta-objective: minimizing correlation in the players' strategies.
The resulting algorithm is still guaranteed to converge to a CCE, and is meta-learned to empirically converge to a Nash equilibrium on a distribution of interest.
If the support of the distribution doesn't include all general-sum games, the problem of finding Nash may be tractable even if P$\neq$PPAD.
We further show this approach is sound by providing a bound on the distance to a Nash equilibrium in terms of our meta-objective.
We evaluate our approach in general-sum imperfect information games.
Our algorithms provide significantly better approximations of Nash equilibria than state-of-the-art regret minimization techniques.

\subsection{Related Work}

The Nash equilibrium is one of the oldest solution concepts in game theory. %\revan{I'm not sure this is true, it's possible pareto-optimality is older}\david{Didn't Nash start the whole thing? I guess it depends on what we call "game theory", but I am happy to say one of the oldest.}
Thanks to its many appealing properties, developing efficient algorithms for approximating Nash equilibria has seen much attention~\citep{kontogiannis2009polynomial,daskalakis2009note,daskalakis2007progress,bosse2010new,deligkas2023polynomial,li2024survey}.
Furthermore, it was shown that, unless P~$=$~NP, polynomial algorithms for finding all Nash equilibria cannot exist~\citep{gilboa1989nash}.
This negative result suggests that there are games for which finding a Nash equilibrium requires enumerating all possible strategies --- an amount exponential in the number of actions.
% \chris{I'm not sure these claims are true. If I understand correctly, P=NP implies efficient algorithms for finding Nash exist, but PPAD-complete problems are not NP-hard, so we can't claim that P $\not =$ NP implies poly-time algorithms for Nash don't exist. Also, even if we could make that claim, I don't know that it's correct to say that no poly-time algorithms existing implies some games require exponential time. Time complexity could be super-polynomial but sub-exponential.}
% \david{I think PPAD is a subclass of NP. Or more precisely, if you can solve certain NP decision problem in poly-time, you could turn it into an efficient algorithm for Nash. Regarding the super-poly stuff, you are right, we can change that. In general, I agree we should be careful, that's why it says "suggests" :) I also added that finding "all Nash equilibria" cannot be done in poly-time, as that would violate the NP-completeness of the decision problem from the beginning of intro.}

The Lemke-Howson algorithm~\citep{lemke1964equilibrium} is one such algorithm, which provably finds a Nash equilibrium of two-player general-sum games in normal-form. 
It works by constructing a path on an abstract polyhedron, which is guaranteed to terminate at the Nash equilibrium.
Similar to the simplex method~\citep{murty1984linear}, the path may be exponentially long in some games.
However, such games are empirically rare~\citep{codenotti2008experimental}.
Several modifications of the Lemke-Howson algorithm were proposed to improve its empirical performance~\citep{codenotti2008experimental,gatti2012combining}.
However, the algorithm cannot work with games in extensive-form.
When converted to normal-form, the size of the game increases exponentially, making these algorithms scale very poorly.

Regret minimization is a powerful framework for online convex optimization \citep{zinkevich2003online, nisan2007algorithmic}, with regret matching as one of the most popular algorithms in game applications \citep{hart2000simple}.
Counterfactual regret minimization enables the use of regret matching in sequential decision-making, by decomposing the full regret to individual states~\citep{zinkevich2008regret}.
In two-player zero-sum games, regret minimization algorithms are guaranteed to converge to a Nash equilibrium. 
Many prior works explored modifications of regret matching to speed up its empirical performance in two-player zero-sum games, such as
CFR$^+$ \citep{tammelin2014solving},
Linear CFR \citep{brown2019deep},
PCFR$^+$ \citep{farina2023regret},
Discounted CFR \citep{brown2019solving},
and their hyperparameter-scheduled counterparts~\citep{zhang2024faster}.

Despite the lack of theoretical guarantees in general-sum games, regret minimization algorithms empirically converge close to Nash equilibria on many standard benchmarks~\citep{risk2010using,gibson2014regret,Pluribus}.
Recently, some theoretical advancements have been made to understand this empirical performance. 
If the game has a special `pair-wise zero-sum' structure, then the regret minimizers are guaranteed to find a Nash equilibrium~\citep{cai2011minimax}. 
Moreover, if a game is `close' to such `pair-wise zero-sum' games, the regret minimzers converge `close' to a Nash equilibrium~\citep{macqueen2023guarantees}.
% \david{do we know of something else?} \revan{It's worthwhile adding a citation to the original zero-sum polymatrix paper}\david{This one \citep{bergman1998separable}?} \revan{I was thinking of the Cai and Daskalakis one}

A recently introduced extension of regret matching, predictive regret matching~\citep{farina2021faster}, forms a continuous class of algorithms with regret minimization guarantees.
Subsequently,~\citep{sychrovsky2024learning} introduced the `learning not to regret' framework---a way to meta-learn the predictions while keeping regret minimization guarantees.
Their aim was to accelerate convergence on a class of oblivious environments.

\subsection{Main Contribution}

%In this work, we extend the learning not to regret framework to general-sum games.
%To encourage convergence to a Nash equilibrium, we penalize correlations in the strategy found by the regret minimizer.
In this work, we extend the \emph{learning not to regret} framework to encourage convergence to Nash equilibria in general-sum games.
Our approach penalizes correlations in the average empirical strategy profile found by the regret minimizer.
While our meta-learned algorithms do not guarantee convergence to a Nash equilibrium, we find that our algorithms empirically converge to CCEs with low correlations in the players' strategies, and provide significantly better approximations of Nash equilibria than prior regret minimization algorithms.

We demonstrate the feasibility of our approach by conducting experiments in multiplayer general-sum games.
We start with a distribution of normal-form games, where prior regret minimization algorithms overwhelmingly converge to a strictly correlated CCE.
Next, we shift our attention to Leduc poker, a standard extensive-form imperfect information benchmark.
We show that, after a small modification of the rules (to make the game general-sum), prior regret minimizers no longer reliably converge to a Nash equilibrium. % in the standard extensive-form benchmark of two-player Leduc poker.
%We show that, after a small modification of the rules of the two-player version of the game, prior regret minimizers no longer reliably converge to a Nash equilibrium.
When trained on this distribution, our meta-learning framework produces a regret minimizer that reach significantly closer to a Nash equilibrium.
Finally, we demonstrate that our framework can even be used to obtain better approximations of a Nash equilibrium on a single general-sum game rather than just a family of games.
We choose the three-player Leduc poker, obtaining, to our best knowledge, the closest approximation of a Nash equilibrium of this game.

\section{Preliminaries}
\label{sec: preliminiaries}

We briefly introduce the formalism of incomplete information games we will use.
Next, we describe regret minimization, a general online convex optimization framework.
Finally, we discuss how regret minimization can be used to find equilibria of these games.

\subsection{Games} 
We work within a formalism based on factored-observation stochastic games~\citep{FOG} with terminal utilities.
% This is without loss of generality as any sequential game with non-terminal utilities is strategically equivalent to a game with only terminal utilities.\footnote{Define the utility at each terminal as the sum of utilities on the path to the terminal in the original game.}

\begin{definition}A game is a tuple
$\langle \mc N, \mc W, w^0, \mc A, \mc T, u, \mc O \rangle$,
where%\vspace{-0.25em}
\begin{itemize}
\itemsep-0.25em 
  \item $\mc N = \{1,\dots, n\}$ is a \defword{player set}. We use symbol $\pl$ for a player and $\opp$ for its opponents.
  \item $\worldstate$ is a set of \defword{world states} and $w^0 \in \worldstate$ is a unique initial world state.
  \item $\jointaction = \jointaction_1 \times \dots \times \jointaction_n$ is a space of \defword{joint actions}. %The subsets $\jointaction_\pl(w) \subset \jointaction_\pl$ and $\jointaction(w) = \jointaction_1(w) \times \jointaction_2(w) \subset \jointaction$ specify the (joint) actions legal at $w \in \mc W$. %$\jointaction_n(w)$ for $n \in \mc N$ are either all non-empty or all empty. 
  A world state with no legal actions is \defword{terminal}. We denote the set of terminal world states as $\mc Z$.
  \item After taking a (legal) joint action $a$ at $w$, the \defword{transition function} $\mc T$ determines the next world state $w'$, drawn from the probability distribution $\mc T(w,a) \in \Delta (\mc W)$.
  \item $u_\pl(z)$ is the \defword{utility} player $\pl$ receives when a terminal state $z\in\mc Z$ is reached. %We denote by $\maxutildiff = \max_{z\in \mc Z} u_\pl(z) - \min_{z\in \mc Z} u_\pl(z)$ the maximum difference in utilities.
  \item $\mc O = (\mc O_1,\dots, \mc O_n)$ %is the \defword{observation function}
  %, where $\mc O_{(\cdot)} : \worldstate \times \jointaction \times \worldstate \to \mb O_{(\cdot)} $ 
  is the \defword{observation function} specifying both the private and public observation that players receive upon the state transition.
\end{itemize}
\end{definition}

The space $\priTree$ of all action-observation sequences can be viewed as the infostate tree of player $\pl$.
A \defword{strategy profile} is a tuple $\vv{\strategy} = (\vv{\strategy}_1,\dots, \vv{\strategy}_n)$, where each player's \defword{strategy} $\vv{\strategy}_\pl : \priHist \in \priTree \mapsto \vv{\strategy}_\pl(\priHist) \in \Delta^{|\jointaction_\pl(\priHist)|}$ specifies the probability distribution from which player $\pl$ draws their next action conditional on having information $\priHist$. 
We denote the space of all strategy profiles as $\strategySet$. A \defword{pure strategy} $\vv{\pstrategy_i}$ is a deterministic strategy: i.e. $\sigma_\pl(\priHist,a_\pl) = 1$ for some $a_i \in \jointaction_\pl(\priHist)$. A selection of pure strategies for all players $\vv{\pstrategy} = (\vv{\pstrategy_1}, \dots \vv{\pstrategy_n})$ is a \defword{pure strategy profile} and the set of all pure strategy profiles is $\vv{\pstrategySet}$.

Let $\Delta(X)$ denote the set of distributions over a domain $X$.  A \defword{joint strategy profile} $\jointstrategy\in \Delta(\pstrategySet)$ is a distribution over pure strategy profiles. As such, every strategy profile is also a joint strategy profile. 
However, the opposite is not true in general: only \emph{some} joint strategy profiles are ``marginalizable'' into an equivalent strategy profile, while those with correlations between players' strategies are not.
% \nathan{$\sigma$ is defined explicitly by $\Delta^{|A_i|}$. We aren't defining the joint strategy profile $\Sigma$ explicitly in this same way (eg $\Delta^{|A|}$).}\david{I don't like it, but this is probably it?}
% \begin{equation*}
%     \jointstrategy:  \bigtimes_{\pl\in\mc N}\mc S_\pl \equiv \bigtimes_{s_\pl\in\mc S_\pl, \pl\in\mc N}\Delta^{|\jointaction_\pl(s_\pl)|}.
% \end{equation*}

% \kevin{Do we not want to define $\Sigma$ as a \textbf{strategy profile distribution} which is $\triangle(\sigma)$?}

The expected \defword{utility} under a joint strategy profile $\jointstrategy$ is $u_\pl(\jointstrategy) = \EX_{z \sim \jointstrategy}{u_\pl(z)}$, where the expectation is over the terminal states $z \in \mc Z$ and their reach probability under $\jointstrategy$. 
The \defword{best-response} to the joint strategy of the other players is $\br(\jointstrategy_{\opp}) \in \argmax_{\strategy_\pl} u_\pl( \strategy_i,\,\jointstrategy_{\opp} )$, where $\jointstrategy_{\opp}(\pstrategy_{\opp}) = \sum_{\pstrategy_\pl\in\jointaction_\pl}\jointstrategy(\pstrategy_\pl, \pstrategy_{\opp})$. 

We may measure the distance of a strategy profile $\strategy$ from a Nash equilibrium by its \defword{NashGap}: the maximum gain any player can obtain by unilaterally deviating from $\strategy$
\begin{equation*}
    {\rm NashGap} (\strategy) = 
    \max_{\pl \in \mc N}
    \left[u_\pl(\br(\strategy_{\opp}), \strategy_\opp) - u_\pl(\strategy)\right].
\end{equation*}
A strategy profile is a Nash equilibrium if its NashGap is zero.\footnote{This is because then the individual strategy profiles are mutual best-responses.} 
% \kevin{Is it true that NashConv is usually defined as a sum instead of a max? If so, should we mention here, maybe in a footnote, that our definition is different in this way?}
% \david{Revan said the same thing. I feel like I have seen both. In 2p zero-sum, I thought exploitability is the sum and NashConv is the max. But idk...}

The coarse correlated equilibrium (CCE)~\citep{moulin1975strategically, nisan2007algorithmic} is a generalization of Nash equilibrium to joint strategy profiles that allows for correlation between players' strategies. 
A CCE is a joint strategy profile such that any unilateral deviation by any player doesn't increase that player's utility, while other players continue to play according to the joint strategy. 
We define the \defword{CCE Gap} as %of a joint strategy profile $\jointstrategy$ as
\begin{equation*}
    {\rm CCE\ Gap} (\jointstrategy) = 
    \max_{\pl \in \mc N}
    \left[u_\pl(\br(\jointstrategy_{\opp}), \jointstrategy_\opp) - u_\pl(\jointstrategy)\right].    
\end{equation*}
% \revan{Should there be a space in CCE Gap? NashGap doesn't have it}\david{I think CCEGap is much harder to read than CCE Gap :/}
A joint strategy profile $\jointstrategy$ is a CCE if and only if its ${\rm CCE\ Gap}$ is non-negative.
If a joint strategy profile has zero CCE Gap, and can be written in terms of its marginal strategies for each player $\jointstrategy = (\strategy_1,\dots,\strategy_n)$, then its marginals $\strategy_\pl$ are a Nash equilibrium.
In general, CCEs do not admit this player-wise decomposition of the joint strategy profile---see Section~\ref{ssec: matrix games} for an example.

\subsection{Regret Minimization}
An \defword{online algorithm} $\bandit$ for the regret minimization task repeatedly interacts with an \defword{environment} through available actions $\actionset_i$.
The goal of a regret minimization algorithm is to maximize its hindsight performance (i.e., to minimize regret). 
%In the regret minimization framework, one is not restricted in the nature of the environment. 
For reasons discussed in the following section, we will describe the formalism from the point of view of player $\pl$ acting at an infostate $s\in \mc S_\pl$.

Formally, at each step $t\le T$, the algorithm submits a \defword{strategy} $\strategy^t_\pl(s) \in \Delta^{|\actionset_\pl(s)|}$. 
Subsequently, it observes the expected \defword{reward} $\reward^t_\pl \in \R^{|\actionset_\pl(s)|}$ at the state $s$ for each of the actions from the environment, which depends on the strategy in the rest of the game.
The difference in reward obtained under $\strategy_\pl^t(s)$ and any fixed action strategy is called the \defword{instantaneous regret} $\regret_\pl(\strategy^t, s) = \reward^t_\pl(\strategy^t) - \inner{\strategy^t_\pl(s)}{\reward^t_\pl(\strategy^t)} \vv{1}$. 
% A sequence of strategies and rewards, submitted by algorithm $\bandit$ and returned by environment $\task$, up to a horizon $T$, is
% \begin{equation}\label{eq:strategy-reward-sequence}
%     \reward^0 \rightarrow \strategy^1 \rightarrow
%     \reward^1 \rightarrow \strategy^2 \rightarrow
%     %\reward^2 \rightarrow \strategy^3 \rightarrow
%     \dots \rightarrow
%     \reward^{T-1}\rightarrow \strategy^{T} \rightarrow
%     \reward^{T},
% \end{equation}
% where we set $\reward^0 = \mathbf{0}$ for notational convenience (see also Figure~\ref{fig: strategy reward sequence}).
The \defword{cumulative regret} throughout time $t$ is $\cumregret^t_\pl(s) = 
    \sum_{\tau=1}^t \regret_\pl(\strategy^{\tau}, s).$
% \begin{equation*}
%     \cumregret^T_\pl = 
%     \sum_{t=1}^T \regret_\pl(\strategy^{t}, \reward^t).
% \end{equation*}

The goal of a regret minimization algorithm is to ensure that the regret grows sublinearly for any sequence of rewards.
One way to do that is for $\bandit$ to select $\strategy^{t+1}_\pl(s)$ proportionally to the positive parts of $ \cumregret^t_\pl(s)$, known as regret matching~\citep{blackwell1956analog}.

\subsection{Connection Between Games and Regret Minimization}
In normal-form games, or when $\privTree$ is a singleton, if the \defword{external regret} $R^{\text{ext}, T}_\pl = \max_{a\in\jointaction_\pl}\cumregret^T_\pl(a)$
grows as $\mathcal{O}(\sqrt{T})$ for all players,
%Then 
then the empirical average joint strategy profile $\overline{\jointstrategy}^T \defeq \frac{1}{T}\sum_{t=1}^T \strategy_1^t \times \dots \times \strategy_n^t$ converges to a CCE as $\mathcal{O}(1/\sqrt{T})$~\citep{nisan2007algorithmic}. 

In extensive-form games, in order to obtain the external regret, we would need to convert the game to normal-form.
However, the size of the normal-form representation is exponential in the size extensive-form representation. 
Thankfully, one can upper-bound the normal-form regret by individual (i.e. per-infostate) \defword{counterfactual regrets}~\citep{zinkevich2008regret}
\begin{equation*}
    \label{eq: cfr bound}
    \sum_{\pl\in\mc N}R^{\text{ext}, T}_\pl \le 
    \sum_{\pl\in\mc N}\sum_{s\in\mc S_\pl} \max\left\{\inftynorm{\cumregret^T_\pl(s)}, 0\right\}.
\end{equation*}
The counterfactual regret is defined with respect to the \defword{counterfactual reward}.
At an infostate $s\in \privTree$, the counterfactual rewards measure the expected utility the player would obtain in the game when playing to reach $s$.
In other words, it is the expected utility of $\pl$ at $s$, multiplied by the opponent's and chance's contribution to the probability of reaching $s$.
We can treat each infostate as a separate environment, and minimize their counterfactual regrets independently.
This approach converges to a CCE~\citep{zinkevich2008regret}.

In two-player zero-sum games, the empirical average strategy $\overline{\strategy}$ is guaranteed to converge to a Nash equilibrium~\citep{zinkevich2008regret}.
In fact, any CCE of a two-player zero-sum game is guaranteed to be marginalizable~\citep{nisan2007algorithmic}.
Intuitively, any correlations will be beneficial for one of the players, which makes it irrational for the opponent to follow it.

% \kevin{And do we want to say here that in 2P0S games, all CCEs are Nashes, so the marginalization is a Nash?}

\begin{algorithm}
    \caption{Neural Predictive Regret Matching~\citep{sychrovsky2024learning}}
    \label{algo:nprm}
 
    \DontPrintSemicolon
    $\cumregret^0 \gets \vv{0} \in \R^{|A|}
    ,\ \reward^0 \gets \vv{0} \in \R^{|A|}$\;
    $\vv{e}_s \gets {\rm embedding\ of\ state\ } s
    $\;
    \Hline{}
    \Fn{\textsc{NextStrategy()}}{
        $\displaystyle\predregret^t \gets [\cumregret^{t-1} + \prediction^{t}]^+$\;
        \textbf{if} $\|\predregret^t\|_1 > 0$\;  
        \hspace{0.2cm}
            \textbf{return} $\strategy^t \gets \predregret^t \ /\ \|\predregret^t\|_1$\;
        % \textbf{else}\; 
        % \hspace{0.2cm}
        \textbf{return} $\strategy^t \gets $ arbitrary point in $\Delta^{|A|}$\hspace*{-1cm}\;
    }
    \Fn{\textsc{ObserveReward(}$\reward^{t}, \vv{e}_s$)}{
        $\cumregret^t \gets \cumregret^{t-1} + \regret(\strategy^{t}, \reward^t)$\;
        $\prediction^{t+1} \gets \alpha(\regret(\strategy^{t}, \reward^t) + \predictor(\regret(\strategy^{t}, \reward^t), \cumregret^t, \vv{e}_s|\theta))$\;
        \vspace{.5mm}
    }
\end{algorithm}

% \begin{algorithm}
% \caption{Neural Predictive Regret Matching \citep{sychrovsky2024learning}}
% \label{algo:nprm}
% \begin{algorithmic}[1]
%     \STATE $\displaystyle \cumregret^0 \gets \vv{0} \in \mathbb{R}^{|A|}, \quad \reward^0 \gets \vv{0} \in \mathbb{R}^{|A|}$
%     \STATE $\displaystyle \vv{e}_s \gets \text{embedding of state } s$
%     \vspace{1em}
%     \STATE NextStrategy()
%         \STATE \hspace{2em}$\displaystyle \predregret^t \gets [\,\cumregret^{t-1} + \prediction^t\,]^+$
%         \STATE \hspace{2em}\textbf{if} {$\|\predregret^t\|_1 > 0$}
%         \STATE \hspace{4em} $\displaystyle \strategy^t \gets \predregret^t / \|\predregret^t\|_1$
%         \STATE \hspace{2em}\textbf{else} 
%         \STATE \hspace{4em} $\strategy^t \gets \text{arbitrary point in } \Delta^{|A|}$
%         \STATE \hspace{2em}\textbf{return} $\displaystyle \strategy^t$
%     \vspace{1em}
%     \STATE ObserveReward($\reward^{t}, \vv{e}_s$)
%         \STATE \hspace{2em}$\displaystyle \regret^t \gets  \regret(\strategy^{t}, \reward^t)$
%         \STATE \hspace{2em}$\displaystyle \cumregret^t \gets \cumregret^{t-1} \;+\; \regret^t$
%         \STATE \hspace{2em}$\displaystyle \prediction^{t+1} \gets \alpha( \regret^t\;+\;\predictor(\regret^t, \cumregret^t, \vv{e}_s \,\big|\theta))$
% \end{algorithmic}
% \end{algorithm}

\section{Meta-Learning Framework}
\label{sec: meta-learning framework}

We aim to find a regret minimization algorithm $\bandit_\theta$ with some parameterization $\theta$ which tends to converge close to a Nash equilibrium on a distribution of games $\tasks$.
%First, we introduce the class of algorithms over which we optimize the convergence.
In this section, we describe the predictive regret minimization algorithm over which we meta-learn.
Then, we formalize our optimization objective for the meta-learning.

\subsection{Neural Predictive Counterfactual Regret Minimization (\NPCFR)}
\label{ssec: npcfr}

We work in the learning not to regret framework~\citep{sychrovsky2024learning}, which is built on the predictive regret matching (PRM)~\citep{farina2021faster}.
% The resulting algorithm enjoys $\mathcal{O}(\sqrt{T})$ bound on the external regret, while optimizing a given meta-objective~\citep{sychrovsky2024learning}.
PRM is an extension of regret matching~\citep{hart2000simple} which additionally uses a predictor about future reward.
PRM provably enjoys $\mathcal{O}(\sqrt{T})$ bound on the external regret for arbitrary bounded predictions~\citep{farina2021faster}.

Neural predictive regret matching is an extension of PRM which uses a predictor $\pi$, parameterized by a neural network $\theta$~\citep{sychrovsky2024learning}; see Algorithm~\ref{algo:nprm}. 
At each step $t$ and each infostate $s\in\mc S_\pl, \pl\in\mc N$, the predictor $\predictor(\cdot|\theta)$ makes a prediction about the next observed regret $\regret^{t+1}$.
This prediction is then used when selecting the strategy, as if that regret was in fact observed.
The strategy is then selected as if this predicted regret was observed. 
Network parameters $\theta$ are shared across all infostates $s\in\mc S_\pl$, $\pl\in\mc N$, and
$\alpha\in\mathbb{R}$ is a hyperparameter, see Appendix~\ref{app: training details} for more details. 
The $\vv{e}_s$ denotes some embedding of the infostate $s$; see Section~\ref{sec: experiments}.

Since we make the predictions bounded, the predictor can be meta-learned to minimize a desired objective while maintaining the regret minimization guarantees~\citep{sychrovsky2024learning}, which makes the algorithm converge to a CCE.
We use a novel meta-objective, which is introduced in the following section, to encourage the algorithm to converge to a Nash equilibrium.
Applying the algorithm to counterfactual regrets at each infostate allows us to use it on extensive-form games.
This setup is refer to as neural predictive counterfactual regret minimization (NPCFR).
% The computational graph of NPCFR is shown in Figure~\ref{fig: comp graph}.

\subsection{Meta-Loss Function} \label{subsec:meta-loss function}

\begin{figure*}
    \centering
    \includegraphics[width=\textwidth]{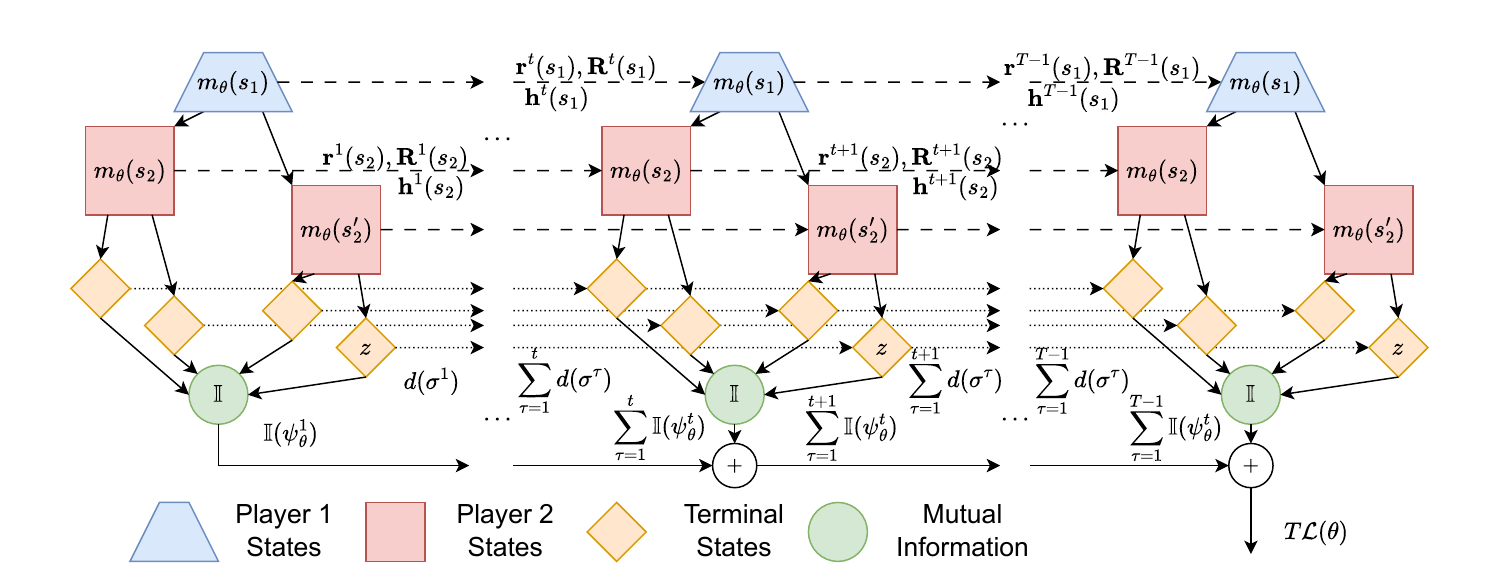}
    \caption{Computational graph of NPCFR$^{(+)}$ for a simple extensive form game. 
    The algorithm $\bandit_\theta$ produces a strategy in each infostate using the regret $\regret^t, \cumregret^t$, and its hidden state $\vv{h}^t$, see Algorithm~\ref{algo:nprm}. 
    Each terminal state $z\in\mc Z$ accumulates its empirical average reach probability $\frac{1}{t}\sum_{\tau=1}^t d(\strategy^\tau)(z)$.
    Marginalizability $\mathbb{I}$ is computed between this accumulated average reach and the reach probability under the empirical average strategy profile in the game tree. 
    The meta-loss is the average mutual information experienced over $T$ steps, according to \eqref{eq: mutual information}. 
    Its gradient is propagated through all edges.}
    \label{fig: comp graph}
\end{figure*}

Any instance of NPCFR is a regret minimizer and is therefore guaranteed to converge to a CCE.
Since any Nash equilibrium is a CCE for which player strategies are uncorrelated, we propose a meta-loss objective that penalizes correlation in the CCE found by NPCFR.  Informally, these correlations measure the mutual dependence of players' strategies. Or in other words, how much a player can infer about the actions of other players given their action. 

One could express this measure of correlation as the \emph{mutual information} of the CCE.\footnote{We describe this measure in more detail in Appendix~\ref{section: proof}.}
However, for extensive-form games, this leads to an exponential blow-up in the size of the game, since there are exponentially more pure strategies than infostates. 
Instead, we exploit the structure of extensive-form games to define an equivalent meta-loss that does not suffer from this blow-up.

% In normal-form games, this means a joint strategy profile that can be written as a product of players' strategies minimizes this meta-loss. To express this meta-loss efficiently in extensive-form games, we define an equivalent meta-loss that exploits the structure of these games.  

% Analogous meta-losses and results for normal-form games are given in Appendix~\ref{section: proof}.

% Analogously, in the extensive-form games we penalize the correlation in terminal reach probabilities. 
% \revan{I think including a bit of normal-form theory could help build intuition. What do you think? }
% \david{I agree, tried to reformulate it a bit :)}

% Formally, let $\psi^T = (\strategy^t)_{t=1}^T$ be a sequence of strategy profiles produced by a regret minimizer.
% Then, in normal-form, the joint strategy profile 
% $\overline{\jointstrategy}^T = \frac{1}{T}\sum_{t=1}^T\bigtimes_{\pl\in\mc N} \strategy_\pl^t$ is guaranteed to have CCE~Gap$(\overline{\jointstrategy}^T)\in \mc O(1/\sqrt{T})$.
% The mutual dependence of actions selected by each agent measures the amount of correlation in the joint strategy profile.
% Informally, these quantify what a player can 

% Regret minimizes produce sequences of strategies. 
Formally, let $\psi^T = (\strategy^t)_{t=1}^T$ be a sequence of strategy profiles selected by a regret minimizer. 
% For any terminal $z\in\mc Z$, 
Let $d(\strategy)$ be the distribution of reach probabilities of terminals $z\in\mc Z$ under $\strategy$, where $d(\strategy)(z)$ is the reach probability of $z$. $d(\strategy)$ can be decomposed into a product of player's (and chance's) contribution of reaching $z$: $d(\strategy)(z) =  d_c(z) \prod_{i \in \mc N} d_i(\strategy)(z)$ where $d_c(z)$ is chance's contribution to reaching $z$ and $d_i(\strategy)(z)$ is the product of $\sigma_\pl(s,a)$ for infostates $s\in\privTree$ on the~path~to~$z$. 
% \revan{maybe we show defined $d_i(\strategy)(z)$ ?}\david{I tried, what do you think?}

The average distribution over terminals across $\psi^T$ is $ d(\psi^T) \defeq \frac{1}{T}\sum_{t=1}^T d(\strategy^t)$. 
We define the \emph{marginal across terminals} $\mu(\psi^T)$ for $\psi^T$ as a distribution across terminals under the empirical average strategy in the game.
Formally, 
\begin{align*}
    \mu(\psi^T)(z) \defeq d_c(z) \prod_{i \in\mc N} \frac{1}{T}\sum_{t=1}^T  d_i(\strategy^t)(z).
\end{align*}
In words, this is the distribution on terminals induced by each player's empirical average strategy in the game tree. The sequence $\psi^T$ is uncorrelated if $d(\psi^T)$ and $\mu(\psi^T)$ have no mutual dependence.  This is formally captured by taking the KL divergence across terminals between $d(\psi^T)$ and $\mu(\psi^T)$. We denote this KL as $\mathbb{I}(\psi)$, since it is equal to mutual information for the two-player case and total correlation for the $n$-player case~\citep{watanabe1960information}. 
\begin{definition}%[Extensive-Form Marginalizability] 
    \label{def: EFM}
    We say that $\psi^T$ is $\epsilon$-extensive-form marginalizable ($\epsilon$-EFM) if
    \begin{equation}
        \label{eq: mutual information}
        \mathbb{I}(\psi^T) \defeq D_{\rm KL}\left(d(\psi^T) \ ||\  \mu(\psi^T) \right)  \leq \epsilon.
    \end{equation}
\end{definition}
When a sequence of strategies of a regret minimizer is close to extensive-form marginalizable, it provably converges close to a Nash equilibrium. Formally, let $\overline{\strategy}^T$ be the average strategy profile~of~$\psi^T$. 
% \begin{theorem}

%     Let $\psi^T$ be a sequence of strategies produced by an external regret minimizer after $T$ iterations.
%     If $\psi^T$ is $\epsilon$-EFM, then the average strategy profile  satisfies
%     \begin{equation*}
%         {\rm NashConv}\left(\overline{\strategy}^T\right) 
%         \le
%         2 M \sqrt{2\epsilon} + \mathcal{O}(1/\sqrt{T}),
%     \end{equation*}
%     where $M = \max_{\pl\in\mc N}\max_{z\in\mc Z} | u_i(z)|$.
% \end{theorem}

\begin{customtheorem}
\label{thm: correctnes of meta-learning}
    If $\psi^T$ was produced by an external regret minimizer with regret bounded by $\mathcal{O}(\sqrt{T})$ after $T$ iterations and  $\psi^T$ is $\epsilon$-EFM, then
    \begin{align}
        \label{eq: marginal thm bound}
        {\rm NashGap} (\overline \strategy^T) \leq \mathcal{O}(1/\sqrt{T})+ 2 M \sqrt{2\epsilon},
    \end{align}
    where $M = \max_{\pl \in \mc N} \max_{z \in \mc Z} | u_i(z)| $.
\end{customtheorem}

% \begin{proof}[Proof Sketch]
%     \david{Revan, please check this sketch and feel free to chance anything}
    
%     In normal-form, the statement can be obtained by using Pinsker's inequality.
%     For an extensive-form game consider its normal-form representation. 
%     This representation is exponentially larger, and so many normal-form strategies have the same extensive-form equivalent.
%     However, some of these strategies introduce additional correlations not present in the extensive-form.
%     Considering the normal-form representation with the minimum mutual information extends the theorem to extensive-form. \kevin{Does it make this explanation simpler if we connect to the concept of ``reduced normal form"?}
%     See Appendix~\ref{app: proof} for the full proof.
% \end{proof}

For a given horizon $T$, we define the meta-loss of NPCFR to be the average mutual information of the average terminal reach of the strategies selected up to $T$ on games $\task\sim \tasks$ %\revan{What is $\theta$ here?}\david{the parameters of the network, since the sequence depends on them. Maybe we should explicitly say that.}
\begin{equation}
    \label{eq: meta-loss}
    \mathcal{L}(\theta) = 
    \mathop{\mathbb{E}}_{\task\in\tasks}
    \left[
    \frac{1}{T}
    \sum_{t=1}^T \mathbb{I}(\psi^t_\theta)
    \right].
\end{equation}
Note minimizing this loss is different from directly minimizing the extensive form marginalizability after $T$ steps. 
We do this to encourage the iterates to be marginalizable as well. 
This is analogous to minimizing $\sum_{t=1}^T f(x^t)$ rather than $f(x^T)$ as in~\citep{andrychowicz2016learning}, where the authors meta-learned a function optimizer.
The computational graph of NPCFR is shown in Figure~\ref{fig: comp graph}.
The gradient of \eqref{eq: meta-loss} originates in the cumulative mutual information and propagates through the game tree, the regrets $\regret^t, \cumregret^t$ and the hidden states $\vv{h}^t$.
The gradients accumulate in the predictor $\pi(\cdot|\theta)$, which is used by the algorithms $\bandit_\theta$ at every information state $s\in\priTree$ and every step $t$, see Algorithm~\ref{algo:nprm}.

\section{Experiments}
\label{sec: experiments}

We conduct our experiments in general-sum games where regret minimizers are not guaranteed to converge to a Nash equilibrium.
Starting in the normal-form setting, we present a distribution of games for which standard regret minimization algorithms converge to a strictly correlated CCE.
We then apply our meta-learning framework to the extensive-form settings, showing we can obtain much better approximate Nash equilibria than prior algorithms.
Finally, we illustrate that the meta-learned algorithms may lose their empirical performance when used out-of-distribution.

We minimize~\eqref{eq: meta-loss} for $T=32$ iterations over $256$ epochs using the Adam optimizer.
The neural network uses two LSTM layers followed by a fully-connected layer.
We performed a small grid search over relevant hyperparameters, see Appendix~\ref{app: training details}.
The meta-learning can be completed in about ten minutes for the normal-form experiments, and ten hours for the extensive-form games on a single CPU. 
See Table~\ref{app: tab: memory requirements} for the memory requirements of all algorithms used.
% For evaluation, we perform $2^{14}=16,384$, resp. $2^{18}=262,144$ iterations for normal, resp. extensive-form games to see how well the algorithms generalize. % outside of the horizon $T$ they were trained on.

We compare the meta-learned algorithms to a selection of current and former state-of-the-art regret minimization algorithms.
Each algorithm is used to minimize counterfactual regret at each infostate of the game tree~\citep{zinkevich2008regret}.
Specifically, we use regret matching (CFR)~\citep{hart2000simple}, predictive regret matching (PCFR)~\citep{farina2021faster}, smooth predictive regret matching (SPCFR)~\citep{farina2023regret}, discounted and linear regret minimization (DCFR, LCFR)~\citep{brown2018solving}, and Hedge~\citep{lattimore2018bandit}.
Whenever applicable, we also investigate the `plus' version of each algorithm~\citep{tammelin2015solving}.

\subsection{Normal-Form Games}
\label{ssec: matrix games}

\setlength{\tabcolsep}{2pt}  % Make columns narrower
\renewcommand{\arraystretch}{1.3}  % Make rows taller
\newcolumntype{P}[1]{>{\hspace{0.5ex}\arraybackslash}p{#1}}

\begin{table*}[t]
\centering
\begin{tabular}{|c||P{5ex}|P{5ex}|P{5ex}|P{5ex}|P{6.5ex}|P{6.5ex}|P{5ex}|P{5ex}|P{5ex}|P{5ex}||P{5ex}|P{5ex}|}
\hline
NashGap & \multicolumn{2}{c|}{CFR$^{(+)}$} & \multicolumn{2}{c|}{PCFR$^{(+)}$} & DCFR & LCFR & \multicolumn{2}{c|}{SPCFR$^{(+)}$} & \multicolumn{2}{c||}{Hedge$^{(+)}$} & \multicolumn{2}{c|}{NPCFR$^{(+)}$} \\ \hline\hline
$10^{-2}$& 0.78 & 0.09 & {\bf 1} & 0.09 & 0.09 & 0.42 & {\bf 1} & 0.09 & {\bf 1} & 0.36 & {\bf 1} & {\bf 1} \\ \hline
$10^{-3}$& 0.09 & 0.02 & 0.91 & 0.02 & 0.02 & 0.02 & {\bf 1} & 0.02 & {\bf 1} & 0.06 & {\bf 1} &  {\bf 1} \\ \hline
$10^{-5}$& 0 & 0 & 0.02  & 0 & 0 & 0 & 0.11 & 0 & 0.25 & 0 & 0.14 & {\bf 1}  \\ \hline
\end{tabular}
\caption{The fraction of games from {\tt biased\_shapley} each algorithm can solve to a given NashGap within $2^{14}=16,384$ steps. 
For the algorithms marked $^{(+)}$, the left column show the standard version, while the right shows the `plus'.
See also Table~\ref{tab: shapley threasholds extended} in Appendix~\ref{app: extre res: shapley}.
}
\label{tab: shapley threasholds}
\end{table*}

The Shapley game
\begin{align}
    %\nonumber
    u_1(\strategy) = \strategy_1^\top \cdot \begin{pmatrix}
        1 & 0 & 0 \\
        0 & 1 & 0 \\
        0 & 0 & 1 \\
    \end{pmatrix}\cdot \strategy_2,
    %\\ 
    \hspace{3ex}
    u_2(\strategy) = \strategy_1^\top \cdot \begin{pmatrix}
        0 & 1 & 0 \\
        0 & 0 & 1 \\
        1 & 0 & 0 \\
    \end{pmatrix}\cdot \strategy_2,
    \label{eq: shapley game}
\end{align}
was used as a simple example where the best-response dynamics doesn't stabilize~\citep{shapley1963some}.
Indeed, it cycles on the elements which are non-zero for one player.
The empirical average joint-strategy converges to a CCE
\begin{equation}
    \label{eq: shapley cce}
    \jointstrategy^* = 
    \frac{1}{6}
    \begin{pmatrix}
        1 & 1 & 0 \\
        0 & 1 & 1 \\
        1 & 0 & 1 \\
    \end{pmatrix}.
\end{equation}
Clearly, $\jointstrategy^*$ is not a Nash equilibrium, as it cannot be written as $\strategy_1 \strategy_2^\top$.
However, thanks to the symmetry of the game, the marginals of $\jointstrategy^*$, or the uniform strategy, turn out to be a Nash equilibrium.

In order to break the symmetry, we perturb the utility of one of the outcomes of the game. 
Specifically, we give payoff $\eta\in\mathbb{R}$ to both players when the first player selects the first, and the second their last action, see Appendix~\ref{app: biased shapley}.
% we modify the game to
% \begin{equation}
%     \label{eq: biased shapley game}
%     u_1(\strategy) = \strategy_1^\top \cdot \begin{pmatrix}
%         1 & 0 & \eta \\
%         0 & 1 & 0 \\
%         0 & 0 & 1 \\
%     \end{pmatrix}\cdot \strategy_2,
%     \hspace{3ex}
%     u_2(\strategy) = \strategy_1^\top \cdot \begin{pmatrix}
%         0 & 1 & \eta \\
%         0 & 0 & 1 \\
%         1 & 0 & 0 \\
%     \end{pmatrix}\cdot \strategy_2,
% \end{equation}
% where $\eta\in\mathbb{R}$.
To preserve that $\jointstrategy^*$ is a CCE, the perturbation $\eta$ needs to be bounded.
We show in Appendix~\ref{app: biased shapley} that for $\eta \le 1/2$, $\jointstrategy^*$ is a CCE.
Furthermore, there is a unique Nash equilibrium, which is non-uniform for $\eta \neq 0$. 
We denote the distribution over biased Shapley games for $\eta\sim\mathcal{U}(a,b)$ as {\tt biased\_shapley}$(a,b)$.

To quantify the performance of the regret minimization algorithms, we study the chance that they find a solution with a given NashGap.
We present our results in Table~\ref{tab: shapley threasholds}.
All the prior regret minimization algorithms fail to reliably find the Nash equilibrium.
The `plus' non-meta-learned algorithms exhibit particularly poor performance in this regime, typically converging to a strictly correlated CCE.
However, they don't all converge to $\jointstrategy^*$ either, see Figure~\ref{fig: heatmaps} for an illustration of the joint strategy profiles each algorithm converges to.
In contrast, NPCFR$^{(+)}$ exhibit fast convergence and remarkable generalization.
We show the convergence comparison of the regret minimization algorithms on {\tt biased\_shapley}$(0,1/2)$ in Figure~\ref{fig: app: biased shapley} in Appendix~\ref{app: extre res: shapley}.
Despite being trained only for $T=32$ steps, our meta-learned algorithms are able to minimize NashGap past $10^4$ steps. 
% In fact, they both exhibit last-iterate convergence up to single-precision accuracy after $\approx4\cdot10^3$ and $\approx5\cdot 10^2$ iterations for \NPCFR\ and \NPCFR$^+$ respectively.

\begin{figure*}
    \centering
    \subfloat[CFR]{{\includegraphics[width=0.15\textwidth]{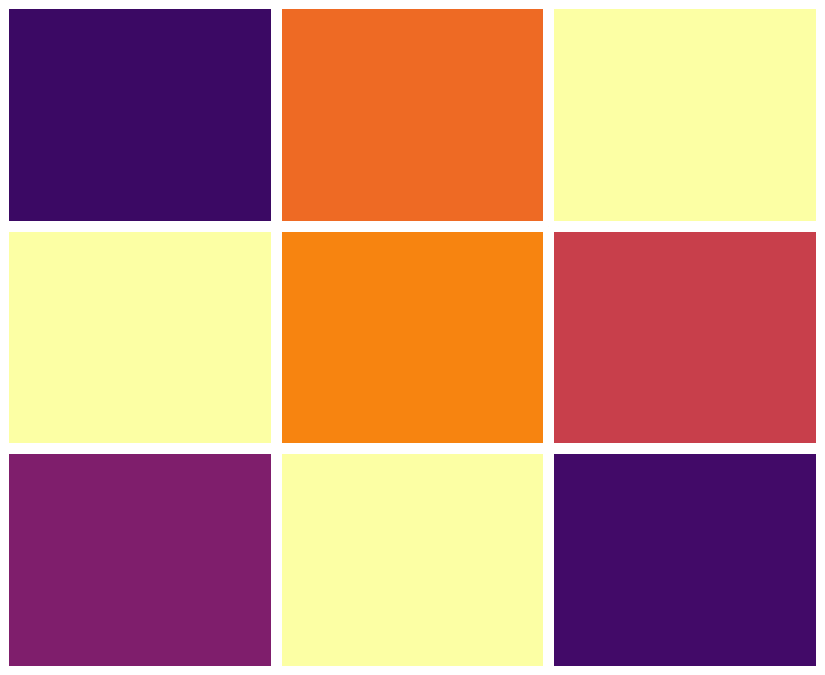}}
                  {\includegraphics[width=0.15\textwidth]{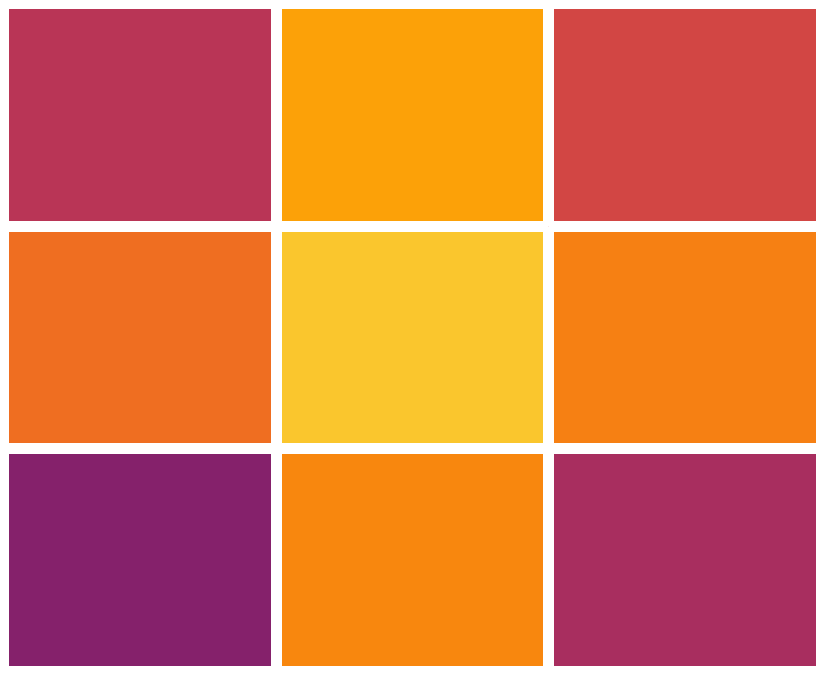}}}\hfil
    \subfloat[CFR+]{{\includegraphics[width=0.15\textwidth]{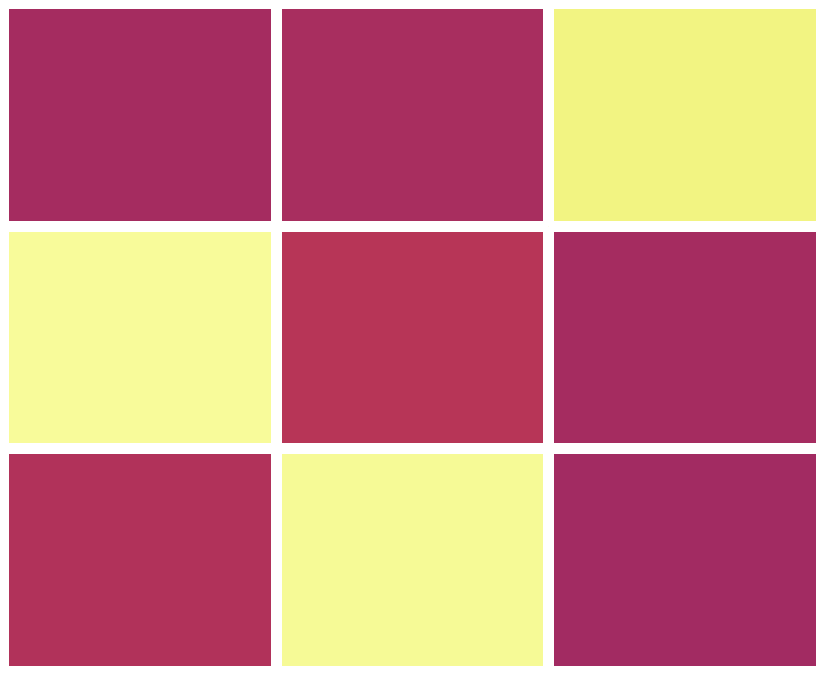}}
                  {\includegraphics[width=0.15\textwidth]{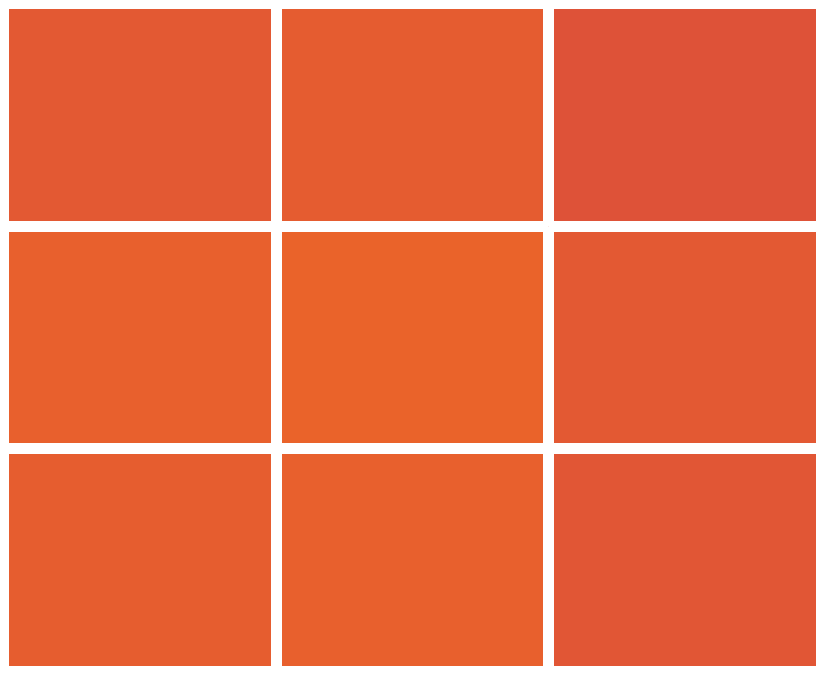}}}\hfil
    \subfloat[PCFR]{{\includegraphics[width=0.15\textwidth]{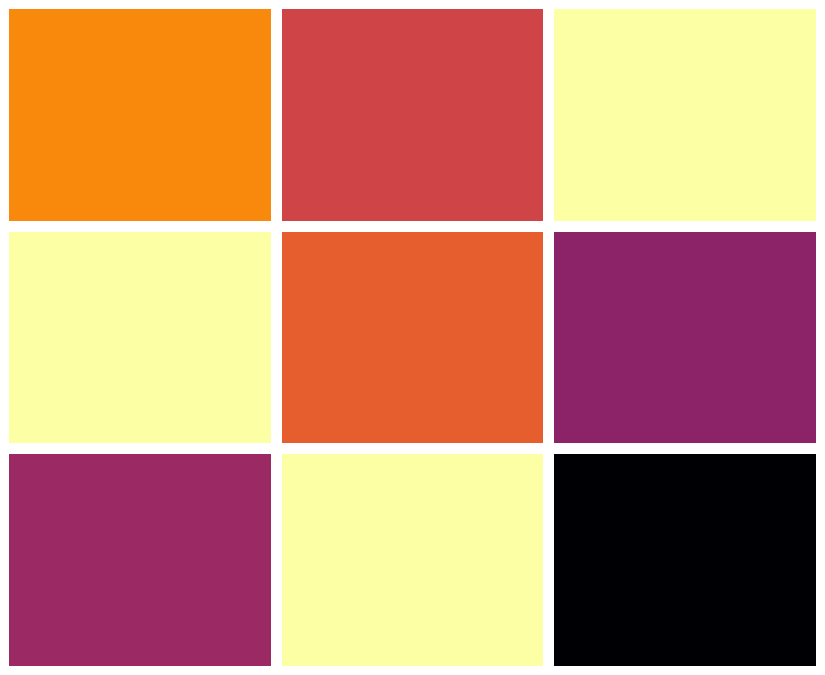}}
                  {\includegraphics[width=0.15\textwidth]{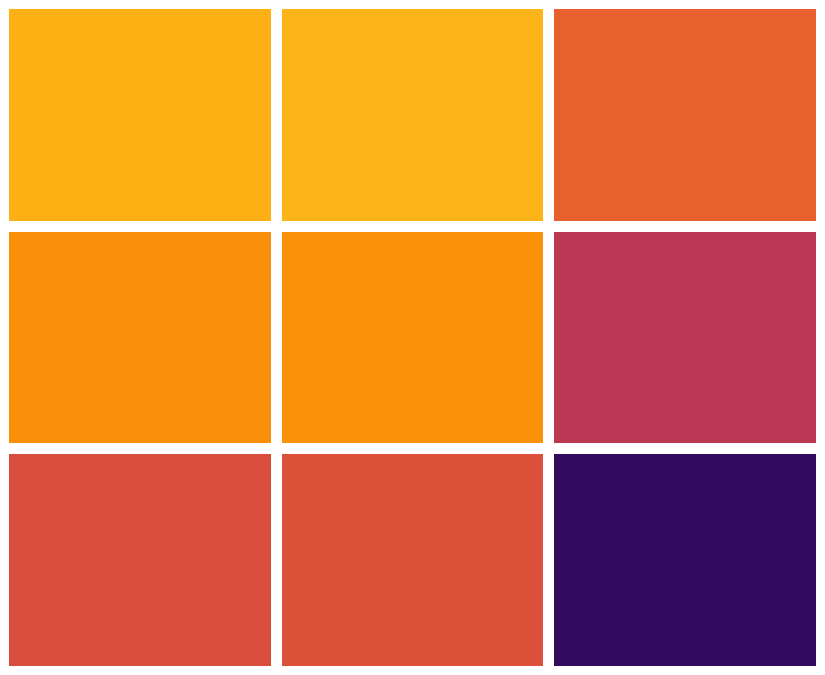}}} \hfil \\

    \subfloat[SPCFR+]{{\includegraphics[width=0.15\textwidth]{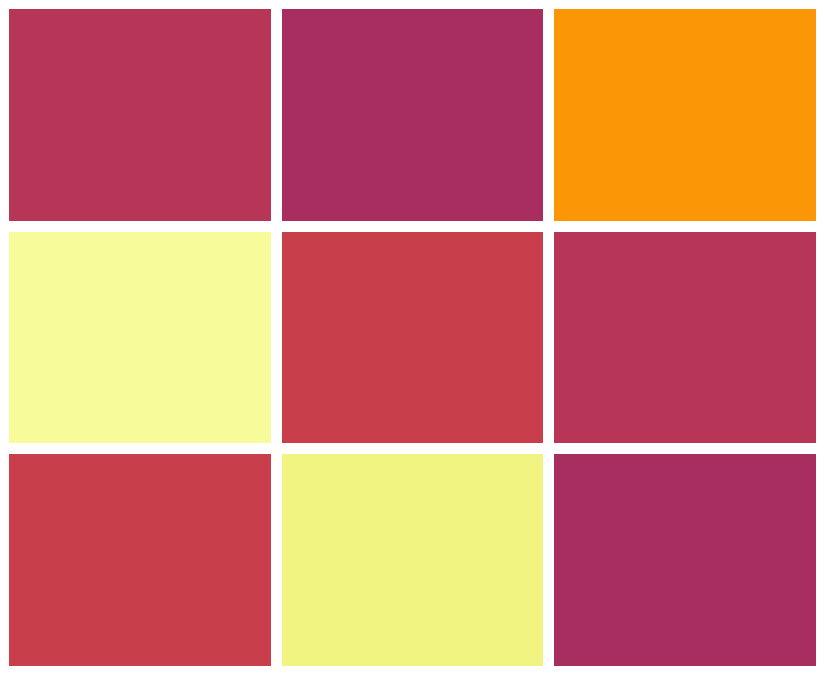}}
                  {\includegraphics[width=0.15\textwidth]{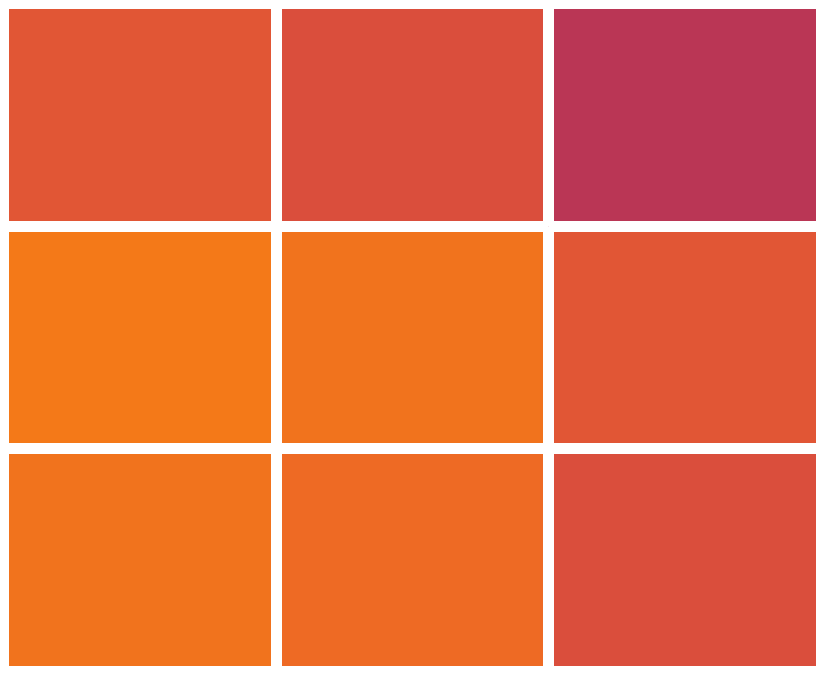}}} \hfil
    \subfloat[NPCFR]{{\includegraphics[width=0.15\textwidth]{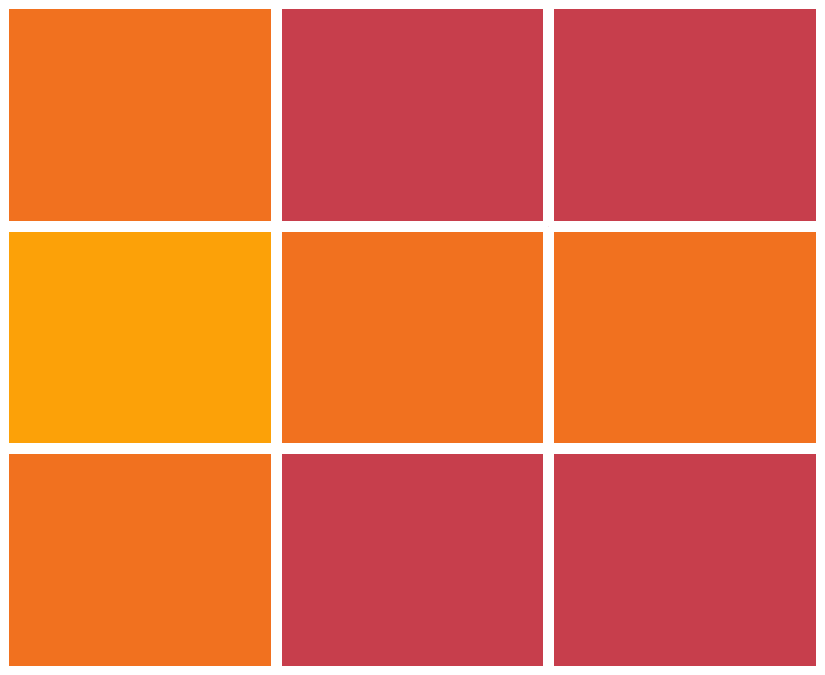}}
                  {\includegraphics[width=0.15\textwidth]{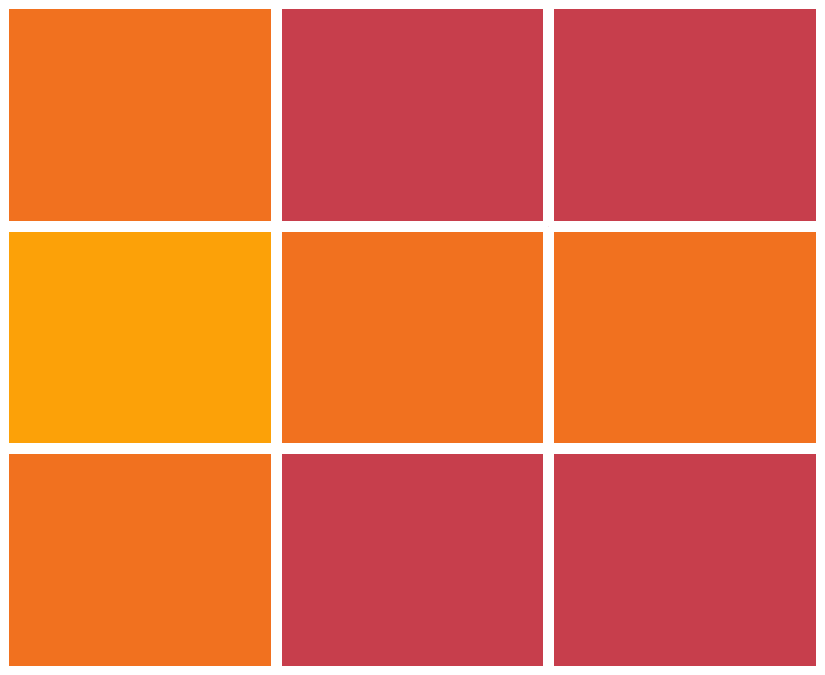}}} \hfil
    \subfloat[NPCFR+]{{\includegraphics[width=0.15\textwidth]{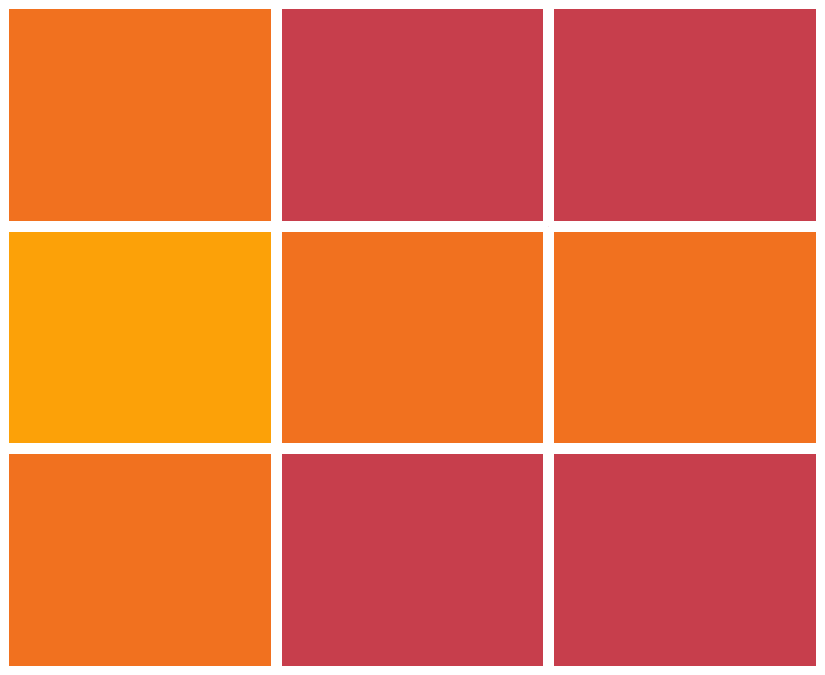}}
                  {\includegraphics[width=0.15\textwidth]{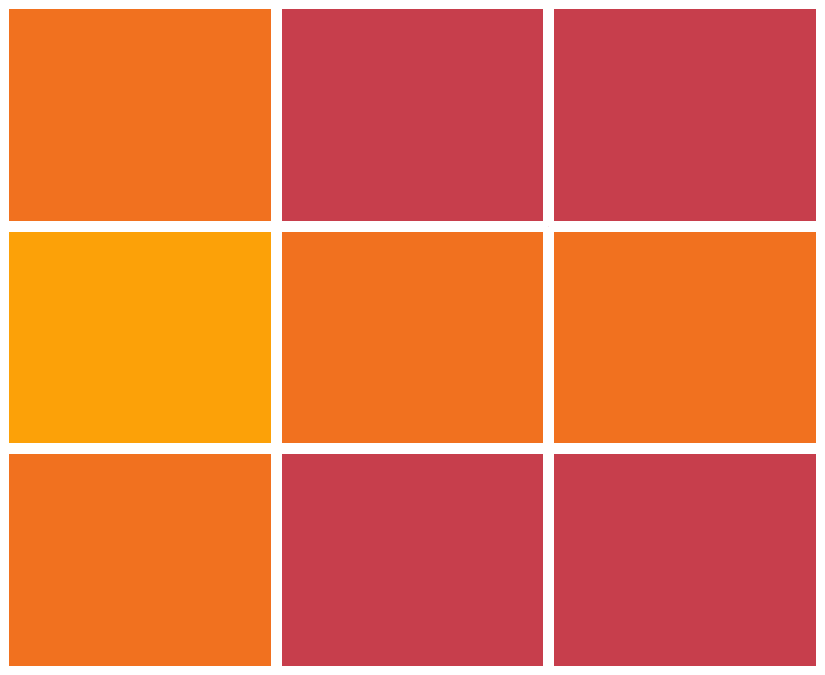}}} \hfil

    \includegraphics[width=0.6\textwidth]{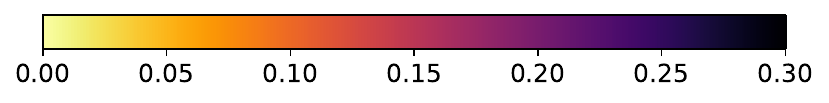}

    \caption{The empirical average joint strategy profiles found by regret minimizers $\overline{\jointstrategy}^T$ (left) and its marginalized version (right) found on a random sample drawn from {\tt biased\_shapley}$(0,1/2)$ after $T=2^{14}$ steps; see Eq.~\eqref{eq: shapley cce}. 
    Darker colors indicate higher probability under $\overline{\jointstrategy}^T$, and minimal differences between left and right figures imply the joint strategy is marginalizable.
    The remaining algorithms are shown in Figure~\ref{fig: app: heatmaps} in Appendix~\ref{app: biased shapley}.}
    \label{fig: heatmaps}
\end{figure*}

\subsection{Extensive-Form Games}
\label{ssec: sequential games}

To evaluate our algorithms in a sequential setting, we use the standard benchmark Leduc poker~\citep{waugh2009strategy}, see Appendix~\ref{app: games: leduc poker} for more details.
% For the three-player version of the game, to reduce the computational requirements, we limit the number of bets per round to one.

\subsubsection{Two-Player Leduc Poker}

% ========= Old results for $\beta\in(0,1)$.========
% \begin{table}[t]
% \centering
% \begin{tabular}{|c||c|c|c|c|c|c|c|c|c|c||c|c|}
% \hline
% NashConv & CFR & CFR$^+$ & PCFR & PCFR$^+$ & SPCFR & SPCFR$^+$ & DCFR & LCFR  & Hedge & Hedge$^+$ & NPCFR & NPCFR$^+$ \\ \hline \hline
% $10^{-2}$ & 0.33 & 1 & 0.42 & 1 & 0.76 & 1 & 0.41 & 0.37 & 0.31 & 0.65 & 0.93 &   \\ \hline
% $10^{-3}$ & 0 & 0.44 & 0.09 & 0.94 & 0.37 & 0.87 & 0 & 0 & 0 & 0.2 & 0.69 &  \\ \hline
% $10^{-5}$ & 0 & 0.19 & 0 & 0.5 & 0 & 0.43 & 0 & 0 & 0 & 0 & 0.69 & \\ \hline
% \end{tabular}
% \vspace{1ex}
% \caption{The percentage of games from {\tt two\_player\_leduc} each algorithm can solve within a given NashConv.}
% \label{tab: steps to reach nashconv}
% \end{table}

Since Leduc poker is a zero-sum game, regret minimizers are guaranteed to converge to a Nash equilibrium in the two-player version.
Under standard rules, players split the pot in the case of a tie, receiving a payoff equal to their total amount bet.
We break the zero-sum property by modifying tie payoffs such that players only receive a $\beta$-fraction of their bets.
This change disincentives betting to increase the size of the pot, but only if the players have the same card ranks, potentially leading to correlations in players' strategies.
% \chris{Is this true? Or does it incentivize correlation in the strategies so that the players can avoid draws?}
% \david{But if you have the same cards, the only way to avoid losing too much is to just not bet, right?}
% \chris{Could the players be ``alternating'' who wins by folding at the end to get a bit more utility overall?}
% \david{I don't think so, that outcome is zero-sum no? So it doesn't make sense for a player to bet only to expect to fold later.}
% \david{Me and Chris are not sure about it. Who knows how it affects Nash. But I still feel like this modification makes it easier to find a CCE, where the players "cheat" and don't bet when they have the same cards.}

We define \twoPlLeduc{} as a distribution over such games, where $\beta\sim\mathcal{U}(0,1/2)$.
To quantify the performance of regret minimization algorithms, we plot the expected NashGap for each algorithm on \twoPlLeduc{} in Figure~\ref{fig: 2 pl leduc}.
While the performance averaged over the domain is similar for all algorithms, the meta-learned algorithms obtain much better approximations of Nash equilibria in each run.
To show this, we investigate the chance that they find a solution with at most a given NashGap.
Table~\ref{tab: 2leduc threasholds} shows the chance for thresholds $10^{-2}, 10^{-3}$, and $10^{-5}$.
With some exceptions, non-meta-learned algorithms generally fail to find a solution with a NashGap of $10^{-2}$.
The `plus' variants perform better empirically but still struggle to obtain solutions close to a Nash equilibrium as reliably as NPCFR$^{(+)}$.
%Similar to the non-meta-learned algorithms, NPCFR$^+$ outperforms NPCFR.
NPCFR$^+$ performs the best overall.

% The change to the rules increases the reach probability of the tie-terminals.
% For the zero-sum version of the game, these states are reach w.p. ???, while in the perturbed setting the chance decreases to ???.\david{this part is not interesting maybe?}

\begin{table*}[t]
\centering
\begin{tabular}{|c||P{5ex}|P{5ex}|P{5ex}|P{5ex}|P{6.5ex}|P{6.5ex}|P{5ex}|P{5ex}|P{5ex}|P{5ex}||P{5ex}|P{5ex}|}
\hline
NashGap & \multicolumn{2}{c|}{CFR$^{(+)}$} & \multicolumn{2}{c|}{PCFR$^{(+)}$} & DCFR & LCFR & \multicolumn{2}{c|}{SPCFR$^{(+)}$} & \multicolumn{2}{c||}{Hedge$^{(+)}$} & \multicolumn{2}{c|}{NPCFR$^{(+)}$} \\ \hline\hline
$10^{-2}$& 0 & {\bf 1} & 0.03 & {\bf 1} &  0.13 & 0 & 0.54 & {\bf 1} & 0 & 0.29 & 0.84 & {\bf 1} \\ \hline
$10^{-3}$& 0 & 0 & 0 & 0.87 &  0 & 0 & 0 & 0.72 & 0 & 0 & 0.73 & {\bf 0.98} \\ \hline
$10^{-5}$& 0 & 0 & 0  & 0.16 & 0 & 0 & 0 & 0.11 & 0 & 0 & 0.73 & {\bf 0.96}  \\ \hline
\end{tabular}
\caption{The fraction of games from \twoPlLeduc{} each algorithm can solve to a given NashGap within $2^{18}=262,144$ steps. 
For the algorithms marked $^{(+)}$, the left column show the standard version, while the right shows the `plus'.
See also Table~\ref{tab: 2leduc threasholds extended} in Appendix~\ref{app: extra res: 2leduc}.
}
\label{tab: 2leduc threasholds}
\end{table*}

\subsubsection{Three-Player Leduc Poker}

%\begin{figure}[t]
%    \centering
%    \includegraphics[width=0.48\textwidth]{figures/3_pl_leduc/nash_conv_3leduc_poker.pdf}
%    \includegraphics[width=0.48\textwidth]{figures/3_pl_leduc/cce_cfr_bound_3leduc_poker.pdf}\\
%    \includegraphics[width=0.8\textwidth]{figures/3_pl_leduc/legend.pdf}
%    \caption{
%    Comparison of non-meta-learned algorithms (\CFR$^{(+)}$, \PCFR$^{(+)}$, \SPCFR$^{(+)}$, Hedge$^{(+)}$, DCFR, and LCFR) with meta-learned algorithms NPRM$^{(+)}$ on {\tt leduc\_poker} (1 bet per round).
   % The figures show NashConv of the average strategy profile $\overline{\strategy}^t$. 
    %Vertical dashed lines separate the training (up to $T=32$ steps) and the %}
    %\label{fig: 3 pl leduc 1 bet}
%\end{figure}

Generally, meta-learning is applied over a distribution of problem instances. 
However, in our setting, it is appealing even to apply it to a single instance of a game.
This is because regret minimization algorithms are not guaranteed to converge to a Nash equilibrium in general-sum games.
However, our meta-learning framework allows us to obtain better approximations of Nash equilibrium.

We demonstrate this approach on the three-player version of Leduc poker; see Appendix~\ref{app: games: leduc poker}.
We refer to the game as {\tt three\_player\_leduc}.
There have been conflicting reports in the literature as to the ability of regret minimization algorithms to converge to a Nash equilibrium in this game~\citep{risk2010using,macqueen2023guarantees}.
% In our experiments, we use a smaller version of the game, allowing only one bet per round; see Appendix~\ref{app: games}.
We found the performance of non-meta-learned algorithms varied significantly, with those using alternating updates giving approx. $4-6$-times better results.
The best approximation of a Nash equilibrium we found among non-meta-learned algorithms using alternating updates\footnote{Among the algorithms we consider, this includes the `plus' algorithms and DCFR. DCFR is similar to CFR$^+$, and was shown to outperform CFR$^+$ on two-player poker~\citep{brown2019solving}.} was ${\rm NashGap} = 0.004$, produced by CFR$^+$.
Without alternating updates, we found ${\rm NashGap} = 0.027$, produced by CFR. 
Our meta-learned algorithms have been able to find a strategy with ${\rm NashGap}=0.012$ for NPCFR, and ${\rm NashGap}=0.001$ for NPCFR$^+$; 
see Table~\ref{tab: 3leduc poker best nashconv} and Figure~\ref{fig: 3 pl leduc} in Appendix~\ref{app: extra res: 3leduc} for details.
To the best of our knowledge, this is the closest approximation of Nash equilibrium of {\tt three\_player\_leduc}.

To the best of our knowledge, the only theoretically sound way to find a Nash equilibrium in this game is to use support-enumeration-based algorithms such as the Lemke-Howson~\citep{lemke1964equilibrium}. 
First, we would need to transform it into a two-player general-sum game.
This can be done by having one of the players always best-respond, and treating them as a part of chance.\footnote{This is the `inverse' of the dummy player argument, which is normally used to show that $n$-player zero-sum games are as hard to solve as $n-1$-player general-sum games.}
However, all of these algorithms work with the game in normal-form.
For {\tt three\_player\_leduc}, the number of pure strategies per player is $\approx 10^{472}$, making these approaches unusable in practice.

% \david{TODO: Investigate what happens if we just do policy gradient with a bunch of compute}

\subsection{Out-of-Distribution Convergence}
\label{ssec: out of dist}

To illustrate that the meta-learned algorithms are tailored to a specific domain, we evaluate them out-of-distribution.
Specifically, we run NPCFR$^{(+)}$, which were trained on {\tt biased\_shapley}$(0,1/2)$, on {\tt biased\_shapley}$(-1,0)$.
When evaluated out-of-distribution, the meta-learned algorithms lose the ability to converge to a Nash equilibrium.
See Figure~\ref{fig: out of dictibution} in Appendix~\ref{app: extra res: out of dist} for more details.

\section{Conclusion}
\label{sec: conclusion}

We present a novel framework for approximating Nash equilibria in general-sum games.
We apply regret minimization, which is a family of efficient algorithms,  guaranteed to converge to a coarse-correlated equilibrium (CCE).
This weaker solution concept allows player to correlate their strategies.
We use meta-learning to search a class of predictive regret minimization algorithms, minimizing the correlations in the CCE found by the algorithm.
The resulting algorithm is still guaranteed to converge to a CCE, and is meta-learned to empirically find close approximations of Nash equilibria.
Experiments in general-sum games, including large imperfect-information games, reveal our algorithms can considerably outperform other regret minimization algorithms.

\paragraph{Future Work.}

Our meta-learning framework might be useful for finding CCEs with desired properties.
For example, one can search for welfare maximizing equilibria by setting the meta-loss to the negative total utility of all players.
We also see other domains, such as auctions, as a promising field where our approach can be used.
One limitation of our approach is that it can be quite memory demanding, especially for larger horizons.
Training on abstractions of the games is promising.

% We can look for CCEs with other interesting properties, such as welfare maximizing?
% Auctions and Nash?

\clearpage
% \section*{Acknowledgments}
% \input{sections/acknowledgements}

%%%%%%%%%%%%%%%%%%%%%%%%%%%%%%%%%%%%%%%%%%%%%%%%%%%%%%%%%%%%%%%%%%%%%%%%

%%% The next two lines define, first, the bibliography style to be 
%%% applied, and, second, the bibliography file to be used.

\bibliographystyle{ACM-Reference-Format} 
\bibliography{main}

%%%%%%%%%%%%%%%%%%%%%%%%%%%%%%%%%%%%%%%%%%%%%%%%%%%%%%%%%%%%

\clearpage
\appendix

\section{Proof of Theorem~\ref{thm: correctnes of meta-learning}} \label{section: proof}
\label{app: proof}

First, we will consider normal-form games and show that if a distribution over strategies is marginalizable (in the sense of Definition~\ref{def: NFM}) and is a CCE, then its marginals form a Nash. This is shown in Theorem~\ref{eq: thm 1 cce}. 

Next, we will show that an alternative form of marginalizability can be used for extensive-form games, given in Definition~\ref{def: EFM}. We call this \emph{extensive-form marginalizability (EFM)}. If this definition is satisfied, it implies that normal-form marginalizability is also satisfied, which is shown in Theorem~\ref{thm: equivalence}.

\subsection{Marginalization in Normal-Form Games}

A \defword{normal-form game} is a tuple $(\mc N, \pstrategySet, u)$ where $\mc N$ is a set of players, $\pstrategySet$ is a set of pure strategy profiles, such that if  $\pstrategySet_i$ is the set of pure strategies for player $i$, we have $\pstrategySet = \bigtimes_{i \in \mc N} \pstrategySet_i$. $u = (u_1, u_2, .., u_n)$ is a tuple of utility functions, where $u_i : \pstrategySet \to \mathbb{R}$.  Players may also randomize their strategies. A \defword{mixed strategy} $s_i$ is a distribution over pure strategies: $s_i \in \Delta(\pstrategySet_i)$ where $\Delta(X)$ is the set of distributions over a domain $X$. 

A joint strategy $\jointstrategy$ is a distribution over \emph{pure strategy profiles}, i.e. $\jointstrategy \in \Delta(\pstrategySet)$.   Let $u_i(\jointstrategy) \defeq \mathbb{E}_{\pstrategy \sim d}u_i(\pstrategy)$ be the expected value under $\jointstrategy$ and $ u_i(\pstrategy_i', \jointstrategy_{\opp}) \defeq \mathbb{E}_{\pstrategy \sim d}u_i(\pstrategy_i', \pstrategy_{\opp})$.

\begin{definition}[CCE]
    A joint strategy $\jointstrategy$ is a CCE if $\forall i, \in \mc N$ and $\pstrategy_i' \in \pstrategySet_i$, we have 
    \begin{align}
     u_i(\pstrategy_i', \jointstrategy_{\opp}) - u_i(\jointstrategy)  \leq 0 .
    \end{align}
\end{definition}

Given some joint strategy $\jointstrategy$, let $\mu_i(\jointstrategy)$ denote the \defword{marginal strategy} of player $i$, where $\mu_i(\jointstrategy)(\pstrategy_{i}) \defeq \sum_{\pstrategy_{\opp} \in \pstrategySet_{\opp}} \jointstrategy(\pstrategy_i, \pstrategy_{\opp})$. Let $\mu(\jointstrategy) = (\mu_1(\jointstrategy), \dots \mu_n(\jointstrategy))$ denote the \defword{ marginal strategy profile}. We will overload $\mu(\jointstrategy)$ to also denote the resulting product distribution over $\pstrategySet$ induced by players playing their marginal strategies. Additionally, let $\jointstrategy_{\opp}$ denote the distribution over $\pstrategySet_{\opp}$ where we marginalize out $i$:
\begin{align}
    \jointstrategy_{\opp} (\pstrategy_{\opp}) \defeq \sum_{\pstrategy_{i} \in \pstrategySet_{i}} \jointstrategy(\pstrategy_i, \pstrategy_{\opp}).
\end{align}
We write $\mu_{\opp}(\jointstrategy_{\opp})$ to denote the marginals of $-i$ from $\jointstrategy_{\opp}$.

\begin{definition}[Normal-Form Marginalizable ] \label{def: NFM}
    We say $\jointstrategy$ is $\epsilon$-normal-form marginalizable ($\epsilon$-NFM) if the total correlation of $\jointstrategy$, denoted $\mathbb{I}(\jointstrategy)$,  is less than $\epsilon$:
    \begin{align*}
        \mathbb{I}(\jointstrategy) \defeq D_{\rm KL}\left( \jointstrategy  \ ||\  \mu(\jointstrategy) \right) \leq \epsilon.
    \end{align*}
\end{definition}

We will show that if a CCE's marginalizability  is small, its marginals are close to a Nash equilibrium. First, we prove a lemma which we will use in the proof of this results.

Let $H(d)$ be the Shannon entropy of a distribution $d$. If $d$ is over a discrete set of outcomes $o$, we have $H(d) \defeq -\sum_{o \in O} d(o) \log(d(o))$.

\begin{lemma}\label{lemma:1}
    Let $\jointstrategy$ be a joint strategy and $\jointstrategy' = (\pstrategy_i', \jointstrategy_{\opp})$ be a  joint strategy where $\pl$ plays pure strategy $\pstrategy_i'$ and $\opp$ play according to  $\jointstrategy_{\opp}$. Then 
    \begin{align}
       \mathbb{I}(\jointstrategy) \geq  \mathbb{I}(\jointstrategy').
    \end{align}
\end{lemma}
\begin{proof}
    By sub-additivity of entropy, we have
    \begin{align*}
        H(\jointstrategy) \leq H(\jointstrategy_{\opp}) + H(\mu_i(\jointstrategy)).
    \end{align*}
    Subtracting $\sum_{i \in N } H(\mu_i(\jointstrategy)) $ from both sides
    \begin{align*}
        & H(\jointstrategy) - \sum_{i \in N } H(\mu_i(\jointstrategy))  \leq H(\jointstrategy_{\opp}) + H(\mu_i(\jointstrategy)) - \sum_{i \in N } H(\mu_i(\jointstrategy)) 
        \\
        \implies &  H(\jointstrategy) - \sum_{i \in N } H(\mu_i(\jointstrategy))  \leq H(\jointstrategy_{\opp})  - \sum_{j \in -i } H(\mu_i(\jointstrategy))
        \\
         \implies &   \sum_{i \in N } H(\mu_i(\jointstrategy)) - H(\jointstrategy)  \geq \sum_{j \in -i } H(\mu_i(\jointstrategy))  -  H(\jointstrategy_{\opp}) .
    \end{align*}
    Next, note that $H(\pstrategy_i') = 0$ since it is a pure strategy (i.e. deterministic distribution). This gives us
    \begin{align*}
         \sum_{i \in N } H(\mu_i(\jointstrategy)) - H(\jointstrategy)  &\geq \sum_{j \in -i } H(\mu_j(\jointstrategy))  -  H(\jointstrategy_{\opp}) +  H(\pstrategy_i') - H(\pstrategy_i')
         \\
         &= \sum_{j \in -i } H(\mu_i(\jointstrategy)) H(\pstrategy_i')  -  H(\jointstrategy'). 
    \end{align*}
    Where the last line uses the fact that $\jointstrategy'$ is a product distribution of $\jointstrategy_{\opp}$ and $\pstrategy_i'$. Using the definition of total correlation, we conclude,
    \begin{align*}
         \mathbb{I}(\jointstrategy) \geq  \mathbb{I}(\jointstrategy').
    \end{align*}
\end{proof}

\begin{theorem}\label{thm: nf}
     If ${\rm CCE\ Gap} (\jointstrategy) = \gamma$ and $\jointstrategy$ is $\epsilon$-normal-form-marginalizable, then ${\rm NashGap} (\mu(\jointstrategy)) \leq \gamma + 2M\sqrt{2 \epsilon} $ where $M = \max_{i \in \mc N} \max_{\pstrategy \in \pstrategySet} | u_i(\pstrategy)|$.
\end{theorem}
\begin{proof}
    First we will show $ \left | u_i(\jointstrategy) - u_i(\mu(\jointstrategy)) \right| \leq  M \sqrt{2 \epsilon}$ for an arbitrary player $i \in \mc N$.
    \begin{align}
       \left | u_i(\jointstrategy) - u_i(\mu(\jointstrategy)) \right| &= \left | \sum_{\pstrategy \in \pstrategySet} u_i(\pstrategy) \jointstrategy(\pstrategy) -  \sum_{\pstrategy \in \pstrategySet} u_i(\pstrategy) \mu(\jointstrategy)(\pstrategy) \right|
       \\
       &=  \left | \sum_{\pstrategy \in \pstrategySet} u_i(\pstrategy) \jointstrategy(\pstrategy) - \mu(\jointstrategy)(\pstrategy) \right| . \label{eq: nf 1}
    \end{align}
    By the triangle inequality, 
    \begin{align}
        \eqref{eq: nf 1} \leq  \sum_{\pstrategy \in \pstrategySet} \left|  u_i(\pstrategy) \jointstrategy(\pstrategy) - \mu(\jointstrategy)(\pstrategy) \right|. \label{eq: nf 2}
    \end{align}
    And using the absolute value, 
    \begin{align}
         \eqref{eq: nf 2} & \leq  \sum_{\pstrategy \in \pstrategySet} | u_i(\pstrategy)|  \left |\jointstrategy(\pstrategy) - \mu(\jointstrategy)(\pstrategy) \right| 
        \\
        & \leq  \sum_{\pstrategy \in \pstrategySet} M  \left |\jointstrategy(\pstrategy) - \mu(\jointstrategy)(\pstrategy)  \right|. \label{eq: nf 3}
    \end{align}

    Since $\mathbb{I}(\jointstrategy) \leq \epsilon$, we may apply Pinsker's inequality, which gives 
    \begin{align*}
         \eqref{eq: nf 2} & \leq   M  \sum_{\pstrategy \in \pstrategySet} \left |\jointstrategy(\pstrategy) - \mu(\jointstrategy)(\pstrategy) \right| \leq 2 M  \sqrt{\epsilon/2} =  M \sqrt{2\epsilon}.
    \end{align*}
    Next, let $\jointstrategy' =  (\pstrategy_i', \mu(\jointstrategy_{\opp}))$  where $\pstrategy_i'$ is an arbitrary pure strategy of $i$. Applying Lemma~\ref{lemma:1} we have
    \begin{align*}
        \mathbb{I}(\jointstrategy') \leq \mathbb{I}(\jointstrategy).
    \end{align*}
    We may then follow identical steps as above with $\jointstrategy'$ instead of $\jointstrategy$. The result is that for any $\pstrategy_i'$,  $\left | u_i(\pstrategy_i', \jointstrategy_{\opp}) - u_i(\pstrategy_i', \mu(\jointstrategy_{\opp})) \right| \leq  M \sqrt{2\epsilon}$. 

    Putting everything together, we will now show that $\mu(\jointstrategy)$ is an approximate Nash equilibrium. First, recall that by assumption
    \begin{align}
        {\rm CCE\ Gap} (\jointstrategy) = 
        \max_{\pl \in \mc N}
        \left[u_\pl(\br(\jointstrategy_{\opp}), \jointstrategy_\opp) - u_\pl(\jointstrategy)\right] = \gamma .
    \end{align}
    Since there is always a pure strategy best-response, we have 
    \begin{align}
        {\rm CCE\ Gap} (\jointstrategy) = 
        \max_{\pl \in \mc N} \max_{\pstrategy_i' \in \pstrategySet_i}
        \left[u_\pl(\pstrategy_i', \jointstrategy_{\opp}) - u_\pl(\jointstrategy)\right] = \gamma.   \label{eq: thm 1 cce} 
    \end{align}
    Using our above results, we may upper bound $ u_i(\jointstrategy)$ with $u_i(\mu(\jointstrategy)) + M\sqrt{2\epsilon}$ and lower bound $ u_i(\pstrategy_i', \jointstrategy_{\opp})$ with $u_i(\pstrategy_i', \mu(\jointstrategy)) - M\sqrt{2\epsilon}$. Thus, 
     \begin{align*}
        \gamma &= \max_{\pl \in \mc N} \max_{\pstrategy_i' \in \pstrategySet_i}
        \left[u_\pl(\pstrategy_i', \jointstrategy_{\opp}) - u_\pl(\jointstrategy)\right]  
        \\
        &\geq  \max_{\pl \in \mc N} \max_{\pstrategy_i' \in \pstrategySet_i}
        \left[(u_i(\pstrategy_i', \mu(\jointstrategy)) -  M\sqrt{2\epsilon}) - (u_i(\mu(\jointstrategy)) + M\sqrt{2\epsilon}) \right] .
    \end{align*}
    Which implies 
    \begin{align*}
    \gamma + 2M\sqrt{2\epsilon} &\geq \max_{\pl \in \mc N} \max_{\pstrategy_i' \in \pstrategySet_i}
    \left[u_i(\pstrategy_i', \mu(\jointstrategy))- u_i(\mu(\jointstrategy)) \right] 
    \\
    &= {\rm NashGap} (\mu(\jointstrategy)),
    \end{align*}
    where we again use the fact there there is a pure strategy best-response. 
\end{proof}

\subsection{Computing Marginalizability In EFGs}

We next show analogous results for extensive-form games, but make use of a definition of marginalizability that exploits the structure of extensive-form games. 

Let $\psi^T = (\strategy^t)_{t=1}^T$ be a sequence of behavior strategy profiles. Recall the definition of $d(\psi^T)$, $\mu(\psi^T)$ and $\mathbb{I}(\psi^T)$ from Section~\ref{subsec:meta-loss function}. We will show that extensive-form marginalizability is equivalent to normal-form marginalizability in guaranteeing the sequence converges to Nash.

% \begin{definition}
%     Given some $n$-player extensive-form game, its \defword{chance-extended induced normal-form} (CE-INF) is a $n+1$ player normal-form game, where chance $c$ is the $n+1$th player. The chance player $c$  has a fixed strategy $s_c$ and cannot deviate. 
% \end{definition}

Given some sequence of mixed strategy profiles  $\zeta^T = (\mstrategy^t)_{t=1}^T$, let $\jointstrategy(\zeta^T)$ denote the empirical distribution of play of $\zeta^T$. Let $\mathbb{I}(\zeta^T)$ be a shorthand for $\mathbb{I}(\jointstrategy(\zeta^T))$. We say a sequence of mixed strategy profiles $\zeta^T$ and behavior strategy profiles $\psi^T$ are \defword{equivalent} if $\forall t\in \{1,\dots T\}, \pl \in \mc N$, we have mixed strategy $\mstrategy_i^t$ produces the contribution over terminals as $\strategy_\pl^t$. This is for any  $\strategy_\pl^t$, the existence of $\mstrategy_i^t$ is guaranteed by Kuhn's theorem in games of perfect recall,  and vice-versa. 

\begin{lemma} \label{lemma: 2}
    For any $\psi^T = (\strategy^t)_{t=1}^T$, there exists a sequence of equivalent mixed strategy profiles $\zeta^T = (\mstrategy^t)_{t=1}^T$ where $\mathbb{I}(\psi^T) = \mathbb{I}(\zeta^T)$.
\end{lemma}

\begin{proof}
    Note that we may associate a single pure strategy profile with each terminal, if we include chance as a player in this pure strategy. Let $\mc N_c$ be the set of players including chance. For this proof, we will denote the set of pure strategy profiles including chance as $\pstrategySet$ and the set of profiles without chance as  $\pstrategySet_{-c}$. Let $m : \mc Z \to \pstrategySet$ be a mapping from terminals to pure strategies where $m(z)$ is a pure strategy that plays to $z$. Let $\pstrategySet^m = \bigcup_{z \in \mc Z} m(z)$. Let $\pstrategySet_i^m = \{ \pstrategy_i \in \pstrategySet_i \mid \exists \pstrategy_{\opp} : (\pstrategy_i, \pstrategy_{\opp}) \in \pstrategySet^m \}$ . We may find a joint strategy $\jointstrategy$ that is equivalent to $d(\psi^T)$ in its resulting distribution over terminals that puts non-zero weight only on $\pstrategySet^m$. Define $\jointstrategy$ as follows:  $\forall z \in \mc Z$, let  $\jointstrategy(m(z)) = d(\psi^T)(z)$. 

    Next, for any $\strategy^t$, let $\mstrategy^t$ be an equivalent mixed strategy profile to $\strategy^t$. Since we must have $\jointstrategy(\pstrategy) = \frac{1}{T}\sum_{t = 1}^T \mstrategy^t(\pstrategy)$  and $\jointstrategy$ only puts positive weight on elements of $\pstrategySet^m$, we have that each  $\mstrategy^t$ only puts positive weight on elements of $\pstrategySet^m$, and for each $i \in \mc N$, $\mstrategy_i^t$ only puts positive weight on $\pstrategySet_i^m$. Thus, if $\equiv$ denotes equivalent in the sense of Kuhn's theorem, we have:
    \begin{align*}
        \strategy_i^t \equiv \mstrategy_i^t \quad \forall i \in \mc N, t \in {1, 2, ..., T}.
    \end{align*}

    Call this sequence of $T$ mixed strategy profiles $\zeta^T$. Let $d_i(\strategy_i^t)$ be the contribution of reaching each terminal by player $i$. Since $\strategy_i^t$ and  $\mstrategy_i^t$ are equivalent for any $i \in N$ and $t \in {1, 2, ..., T}$, we have
    \begin{align*}
        \frac{1}{T} \sum_{t = 1}^T d_i(\strategy_i^t)(z) =  \frac{1}{T} \sum_{t = 1}^T d_i(\mstrategy_i^t)(m(z)) \quad \forall z \in \mc Z.
    \end{align*}
    Which implies 
    \begin{align*}
        \prod_{i \in \mc N_c} \frac{1}{T} \sum_{t = 1}^T d_i(\strategy_i^t)(z) &= \prod_{i \in \mc N_c}  \frac{1}{T} \sum_{t = 1}^T d_i(\mstrategy_i^t)(m(z))  \quad \forall z \in \mc Z
        \\
        \implies  \mu(\psi^T)(z) &=  \mu(\zeta^T)(m(z))  \quad \forall z \in \mc Z.
    \end{align*}    

    Next, returning to $\mathbb{I}(\psi^T) $, we have:
    \begin{align*}
        \mathbb{I}(\psi^T) \defeq \sum_{z \in \mc Z} d(\psi^T)(z) \log \left( \frac{ d(\psi^T)(z)}{ \mu(\psi^T)(z)}  \right). 
    \end{align*}

    Replacing the summation over $z \in \mc Z$ to one over $\pstrategy \in \pstrategySet^m$, we have:
    \begin{align}
        \mathbb{I}(\psi^T) = \sum_{\pstrategy \in \mc \pstrategySet^m} d(\psi^T)(m^{-1}(\pstrategy)) \log \left( \frac{ d(\psi^T)(m^{-1}(\pstrategy))}{ \mu(\psi^T)(m^{-1}(\pstrategy))}  \right).  \label{eq: replace}
    \end{align}
    Where $m^{-1}$ is the inverse mapping of $m$. Note that $m^{-1}$ is injective from  $\pstrategySet^m$ to $\mc Z$, since each $z$ is produced by a unique sequence of actions including chance. 

    By construction of $\zeta^T$, $\jointstrategy(\zeta^T)(m(z)) = d(\psi^T)(z) \ \forall z \in \mc Z$, which implies $ \jointstrategy(\zeta^T)(\pstrategy) = d(\psi^T)(m^{-1}(\pstrategy))$. Likewise, since $ \mu(\psi^T)(z) =  \mu(\zeta^T)(m(z)) \ \forall z \in \mc Z $, we have $ \mu(\psi^T)(m^{-1}(\pstrategy)) =  \mu(\zeta^T)(\pstrategy)$. Using these substitutions into \eqref{eq: replace}, we have 
     \begin{align*}
        \mathbb{I}(\psi^T) = \sum_{\pstrategy \in \mc \pstrategySet^m} \jointstrategy(\zeta^T)(\pstrategy) \log \left( \frac{ \jointstrategy(\zeta^T)(\pstrategy)}{ \mu(\zeta^T)(\pstrategy)}  \right) \defeq  \mathbb{I}_c(\zeta^T). 
    \end{align*}
    Where $ \mathbb{I}_c(\zeta^T)$ denotes the normal-form marginalizability where we include chance as a player. Let  $\mathbb{I}(\zeta^T)$ denote the normal-form marginalizability where  where we \emph{do not} include chance as a player. Then, we have
    \begin{align*}
        \mathbb{I}_c(\zeta^T)  \defeq \sum_{\pstrategy \in \mc \pstrategySet^m} \jointstrategy(\zeta^T)(\pstrategy) \log \left( \frac{ \jointstrategy(\zeta^T)(\pstrategy)}{ \mu(\zeta^T)(\pstrategy)}  \right) .
    \end{align*}

    However, note that since chance always plays a fixed strategy, which we denote $\mstrategy_c$, we have
    \begin{align*}
        \jointstrategy(\zeta^T)(\pstrategy)  &= \frac{1}{T} \sum_{t=1}^T \left( \prod_{i \in N_c} \mstrategy^t_i(\pstrategy_i) \right)
        \\
        & = \mstrategy_c(\pstrategy_c) \frac{1}{T}\sum_{t=1}^T   \left(  \prod_{i \in N} \mstrategy^t_i(\pstrategy_i) \right)
        \\
        &=  \mstrategy_c(\pstrategy_c) \frac{1}{T}\sum_{t=1}^T   \left(  \prod_{i \in N} \mstrategy^t_i(\pstrategy_i) \right) 
        \\
        &= \mstrategy_c(\pstrategy_c) \jointstrategy_{-c}(\zeta^T)(\pstrategy_{-c}),
    \end{align*}
    where $\jointstrategy_{-c}(\zeta^T)$ is a joint strategy for all players except $c$. 
    
    Similarly, we have
    \begin{align*}
        \mu(\zeta^t)(\pstrategy) &= \prod_{i \in N_c} \frac{1}{T}\sum_{t=1}^T     \mstrategy^t_i(\pstrategy_i)
        \\
        &=  \left(\frac{1}{T}\sum_{t=1}^T     \mstrategy^t_c(\pstrategy_c) \right)   \prod_{i \in N} \frac{1}{T}\sum_{t=1}^T     \mstrategy^t_i(\pstrategy_i)
        \\
        &= \mstrategy_c(\pstrategy_c)    \prod_{i \in N} \frac{1}{T}\sum_{t=1}^T     \mstrategy^t_i(\pstrategy_i)
        \\
        &=  \mstrategy_c(\pstrategy_c)  \mu_{-c}(\zeta^T)(\pstrategy_{-c}),
    \end{align*}
    where $ \mu_{-c}(\zeta^T)$ are the marginals of  $\jointstrategy_{-c}(\zeta^T)$ for all players except chance.  

    Then, we have
    \begin{align*}
        \mathbb{I}_c(\zeta^T)  &\defeq \sum_{\pstrategy \in \mc \pstrategySet^m} \jointstrategy(\zeta^T)(\pstrategy) \log \left( \frac{ \jointstrategy(\zeta^T)(\pstrategy)}{ \mu(\zeta^T)(\pstrategy)}  \right) 
        \\
        & =\sum_{\pstrategy \in \mc \pstrategySet^m_{-c}} \sum_{\rho_c \in \pstrategySet^m_{c}}  
         \mstrategy_c(\pstrategy_c) \jointstrategy_{-c}(\zeta^T)(\pstrategy_{-c})
         \\ &\hspace{14ex} \log \left( \frac{  \mstrategy_c(\pstrategy_c) \jointstrategy_{-c}(\zeta^T)(\pstrategy_{-c})}{ \mstrategy_c(\pstrategy_c)  \mu_{-c}(\zeta^T)(\pstrategy_{-c})}  \right)
        \\
        & =\sum_{\pstrategy \in \mc \pstrategySet^m_{-c}} \sum_{\rho_c \in \pstrategySet^m_{c}}  
         \mstrategy_c(\pstrategy_c) \jointstrategy_{-c}(\zeta^T)(\pstrategy_{-c})
         \\ & \hspace{14ex} \log \left( \frac{  \jointstrategy_{-c}(\zeta^T)(\pstrategy_{-c})}{ \mu_{-c}(\zeta^T)(\pstrategy_{-c})}  \right)
         \\
           & = \sum_{\pstrategy \in \mc \pstrategySet^m_{-c}}  \jointstrategy_{-c}(\zeta^T)(\pstrategy_{-c})
         \log \left( \frac{  \jointstrategy_{-c}(\zeta^T)(\pstrategy_{-c})}{ \mu_{-c}(\zeta^T)(\pstrategy_{-c})}  \right) 
         \\ & \hspace{14ex} \sum_{\rho_c \in \pstrategySet^m_{c}}  
         \mstrategy_c(\pstrategy_c)
         \\
          & = \sum_{\pstrategy \in \mc \pstrategySet^m_{-c}}  \jointstrategy_{-c}(\zeta^T)(\pstrategy_{-c})
         \log \left( \frac{  \jointstrategy_{-c}(\zeta^T)(\pstrategy_{-c})}{ \mu_{-c}(\zeta^T)(\pstrategy_{-c})}  \right) 
         \\
        &= \mathbb{I}(\zeta^T),
    \end{align*}
    where the second last line uses the fact that $\sum_{\rho_c \in \pstrategySet^m_{c}}  \mstrategy_c(\pstrategy_c) = 1$.  Thus we have shown that $\mathbb{I}(\psi^T) = \mathbb{I}(\zeta^T)$.
\end{proof}

Let $\overline \strategy$ be the average behavior strategy profile of $\psi$ (for a definition see equation 4 of \citep{zinkevich2008regret}). 

The previous lemma establishes a relationship between extensive-form marginalizability and normal-form marginalizability. This allows us to show an analogous result to Theorem~\ref{thm: nf} for extensive-form games. It shows that if $\psi^T$ is produced by a regret-minimizer after $T$ iterations and $\psi^T$ is approximately extensive-form marginalizable, then  $\overline \strategy^T$ is boundedly far from Nash.

\begin{customtheorem}
\label{thm: equivalence}
    If $\psi^T$ was produced by an external regret minimizer with regret bounded by $\mathcal{O}(\sqrt{T})$ after $T$ iterations and  $\psi^T$ is $\epsilon$-EFM, then
    \begin{align}
        {\rm NashGap} (\overline \strategy^T) \leq \mathcal{O}(1/\sqrt{T})+ 2 M \sqrt{2\epsilon},
    \end{align}
    where $M = \max_{\pl \in \mc N} \max_{z \in \mc Z} | u_i(z)| $.
\end{customtheorem}
\begin{proof}
    Let $\zeta^T$ be an equivalent sequence of mixed strategies to $\psi^T$ as defined by Lemma~\ref{lemma: 2}. 
    If $\psi^T$ has been produced by a regret minimizer, the empirical distribution of play, $\jointstrategy(\zeta^T)$, is a $\mathcal{O}(1/\sqrt{T})$-CCE. 
    Applying Theorem~\ref{thm: nf} shows that the NashGap of marginals $\mu(\jointstrategy(\zeta^T))$ is  $\mathcal{O}(1/\sqrt{T})+ 2 M' \sqrt{2\epsilon}$ where $M'= \max_{\pl \in \mc N} \max_{\pstrategy \in \pstrategySet} | u_i(\pstrategy)|$. 
    Since the marginals $\mu(\jointstrategy(\zeta^T))$ are equivalent in their distribution over terminals to  $\overline \strategy^T$,  the NashGap of $\overline \strategy^T$ must also be a $\mathcal{O}(1/\sqrt{T})+ 2M' \sqrt{2\epsilon}$. Lastly, we may replace $M'$ with $M = \max_{\pl \in \mc N} \max_{z \in \mc Z} | u_i(z)|$ since $M \geq M'$. 
\end{proof}

\section{Training Details}
\label{app: training details}
\begin{figure*}[t!]
    \centering
    \includegraphics[width=0.8\textwidth]{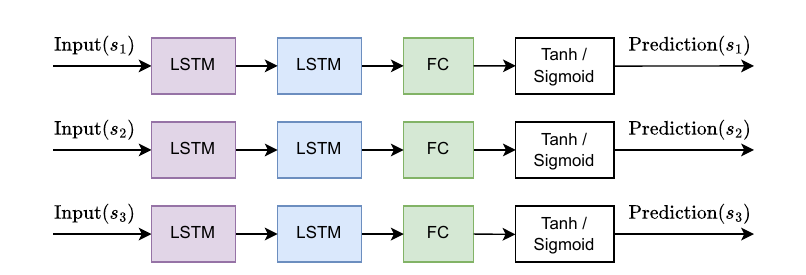}
    \caption{Neural network architecture used for all meta-learned algorithms. 
    The input is processed by two LSTM layers and a fully connected (FC) layer, treating each infostate $\{s_k\}_{k=1}^3$ as a separate batch element.
    We bound the output of the network by using the Tanh, or Sigmoid activation function.
    This choice is treated as a hyperparameter.
    Same colours indicate shared parameters.}
    \label{fig: network architecture}
\end{figure*}

In order to select values for the size of the network, batch size, $l_2$ weight decay, and the prediction scaling $\alpha$, we performed a small grid search.
In order to choose between Tanh and Sigmoid activation at the last layer of the network (see Figure~\ref{fig: network architecture}), we found that looking at the performance of PCFR$^{(+)}$ is a good indicator.
If the performance of PCFR$^{(+)}$ on the domain of interest is good, using Tanh should be the default option. 
This is because it is easy for the network to imitate PCFR$^{(+)}$ by making the prediction zero (assuming $\alpha=1$).
On the other hand, making a negative prediction can result in the action receiving zero probability, causing the meta-gradient to be zero for that prediction. 
This can lead to instabilities, especially on shorter training instances.

We found that, in {\tt biased\_shapley}, the value of the prediction scaling has major impart on the performance of NPCFR.
For small values of $\alpha$, the algorithm would accumulate large negative regret, forcing it to follow the best-response dynamics.
We found that for $\alpha \ge 2$ these oscillations would disappear, as the algorithm `looks sufficiently far into the future'.
Making negative predictions seems to hinder the algorithm, and we thus opt to use Sigmoid in the network for {\tt biased\_shapley}.

On Leduc poker, we observed good empirical performance of PCFR$^{(+)}$.
Following the above arguments, we used Tanh as the activation of the prediction network $\pi(\cdot|\theta)$.
We also found that large values of $\alpha$ may lead to instability.
For NPCFR$^+$ on {\tt three\_player\_leduc}, we modified the form of the prediction, see Appendix~\ref{app: extra res: 3leduc} for more details.

\section{Games}
\label{app: games}
\subsection{Biased Shapley}
\label{app: biased shapley}

The biased Shapley game is given by
\begin{equation*}
    u_1(\strategy) = \strategy_1^\top \cdot \begin{pmatrix}
        1 & 0 & \eta \\
        0 & 1 & 0 \\
        0 & 0 & 1 \\
    \end{pmatrix}\cdot \strategy_2,
    \hspace{1ex}
    u_2(\strategy) = \strategy_1^\top \cdot \begin{pmatrix}
        0 & 1 & \eta \\
        0 & 0 & 1 \\
        1 & 0 & 0 \\
    \end{pmatrix}\cdot \strategy_2.
\end{equation*}

\paragraph{Course-Correlated Equilibrium}
We want to determine the values of $\eta$ for which the joint strategy
\begin{equation*}
    \jointstrategy^* = 
    \frac{1}{6}
    \begin{pmatrix}
        1 & 1 & 0 \\
        0 & 1 & 1 \\
        1 & 0 & 1 \\
    \end{pmatrix},
\end{equation*}
is a CCE.
We focus on player one, and note that the discussion involving player two leads to similar conclusion after permuting the order of the actions.

The expected payoff of player one under $\jointstrategy^*$ is $1/2$.
Deviating to either the second or last action doesn't increase the payoff.
Deviating to the first action leads to expected payoff
\begin{equation*}
    u_1' = 1/3 + \eta / 3.
\end{equation*}
In order for $\jointstrategy^*$ to be a CCE we thus need
\begin{equation*}
    1/2 \ge 1/3 + \eta / 3 
    \hspace{3ex}\Leftrightarrow\hspace{3ex}
    \eta \le 1/2.
\end{equation*}

\paragraph{Nash equilibrium}
We will use support enumeration to find all Nash equilibria of the biased Shapley game.
Support enumeration requires the game to be non-degenerate, or for each strategy to have at most as many best-responses as the support size of that strategy.
Since the best-response in the biased Shapley game is unique, this requirement is satisfied.

For support size one, since the best-response dynamics cycle, there is no pair of strategies which are mutual best-responses.
Similarly, any strategy of support size two will have an action for at least one player in the best-response cycle, which is not in the support of those strategies.
It thus remains that the support size must be three, and the Nash equilibrium is unique.

To get the exact Nash equilibrium, each player needs to select a strategy which offers the same payoff to the other player.
The strategy of player one would thus need to satisfy
\begin{equation*}
    \sigma_{1,1} =
    \sigma_{1,2} + \eta \sigma_{1,1} =
    \sigma_{1,3} = 
    1 - \sigma_{1,1} - \sigma_{1,2},
\end{equation*}
which gives
\begin{align*}
    \sigma_{1,2} + (\eta-1) \sigma_{1,1} &=
    0,\\
    (3-\eta) \sigma_{1,1} &=
    1 .
\end{align*}
or
\begin{equation*}
    \strategy_1^\top =
    \left(
    \frac{1}{3-\eta}, \frac{1-\eta}{3-\eta}, \frac{1}{3-\eta}
    \right).
\end{equation*}
Similarly, we could obtain the Nash equilibrium strategy for player two
\begin{equation*}
    \strategy_2^\top =
    \left(
    \frac{1-\eta}{3-\eta}, \frac{1}{3-\eta}, \frac{1}{3-\eta}
    \right).
\end{equation*}

\subsection{Leduc Poker}
\label{app: games: leduc poker}

Leduc poker is a simplified version of Texas Hold'em poker.
For the $n$-player version, the deck consists of $2(n + 1)$ cards of two suits and $n + 1$ ranks.
At the beginning of the game, each player is privately dealt one card, and puts $\$1$ to the pot.
The game proceeds in two betting rounds.
In each round, players take turns in a given order, having the option to check, bet, call a previous bet, or fold.
In each round, there can be at most two consecutive bets.
The bet size is $\$2$ and $\$4$ in the first and second round respectively.
A public card is revealed after the first round.

Among the players who didn't fold, the player who's private card has the same rank as the public card wins.
If there is no such player, the highest card wins.
Finally, if there are two players with the highest rank card and no pair it is a tie, they split the pot equally.

% Note that unlike many official implementations of the three-player version of Leduc poker, we only allow one bet in each round.
% This reduces the size of the game from 25,800 to 4,128 infostates.
% We opted for this reduction, since the meta-learning algorithm needs to keep the whole computational graph of $T$ steps in memory.
% For the full version of the game, these memory requirements become rather substantial even for relatively small values of $T$.

\section{Additional Experimental Results}
\label{app: extra experiments}
\subsection{Biased Shapley}
\label{app: extre res: shapley}

\begin{figure*}[t]
    \centering
    \includegraphics[width=0.48\textwidth]{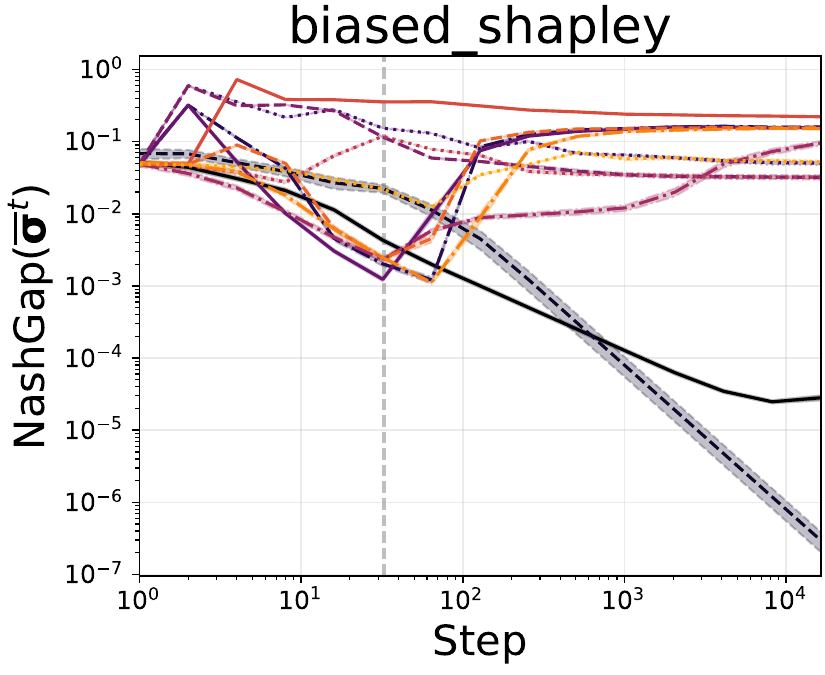}
    \includegraphics[width=0.48\textwidth]{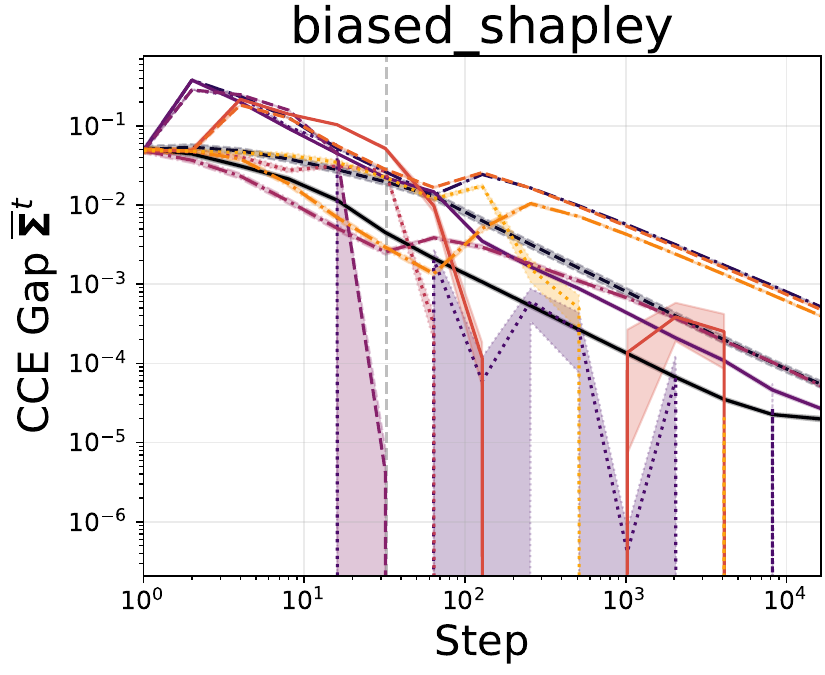}
    
    \includegraphics[width=0.8\textwidth]{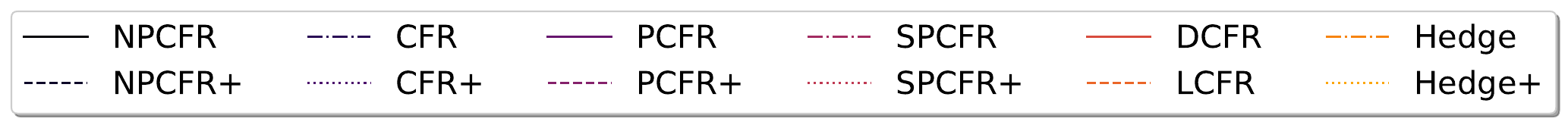}
    \caption{
    Comparison of non-meta-learned algorithms (\CFR$^{(+)}$, \PCFR$^{(+)}$, \SPCFR$^{(+)}$, Hedge$^{(+)}$, DCFR, and LCFR) with meta-learned algorithms NPCFR$^{(+)}$. NPCFR$^{(+)}$ was trained on {\tt biased\_shapley}$(0,1/2)$.
    The figures show NashGap of the average strategy profile $\overline{\strategy}^t$. 
    Vertical dashed lines separate the training (up to $T=32$ steps) and the generalization (from $T$ to $2^{14}=16,384$ steps) regimes. 
    }
    \label{fig: app: biased shapley}
\end{figure*}

\begin{figure*}
    \centering
    \subfloat[CFR]{{\includegraphics[width=0.15\textwidth]{figures/heatmaps/CFR.pdf}}
                  {\includegraphics[width=0.15\textwidth]{figures/heatmaps/CFR-decorrelated.pdf}}}\hfil
    \subfloat[CFR+]{{\includegraphics[width=0.15\textwidth]{figures/heatmaps/CFR+.pdf}}
                  {\includegraphics[width=0.15\textwidth]{figures/heatmaps/CFR+-decorrelated.pdf}}}\hfil
    \subfloat[DCFR]{{\includegraphics[width=0.15\textwidth]{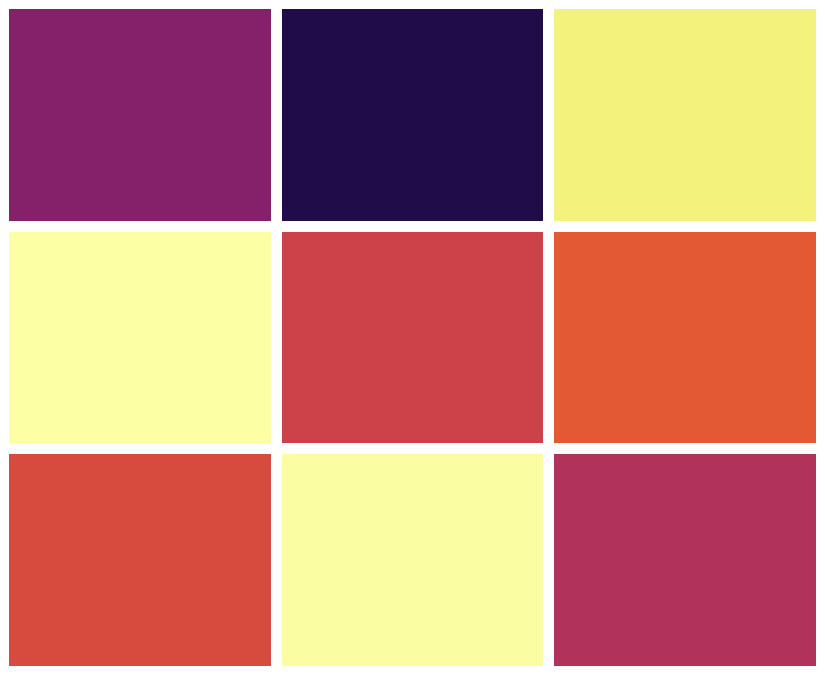}}
                  {\includegraphics[width=0.15\textwidth]{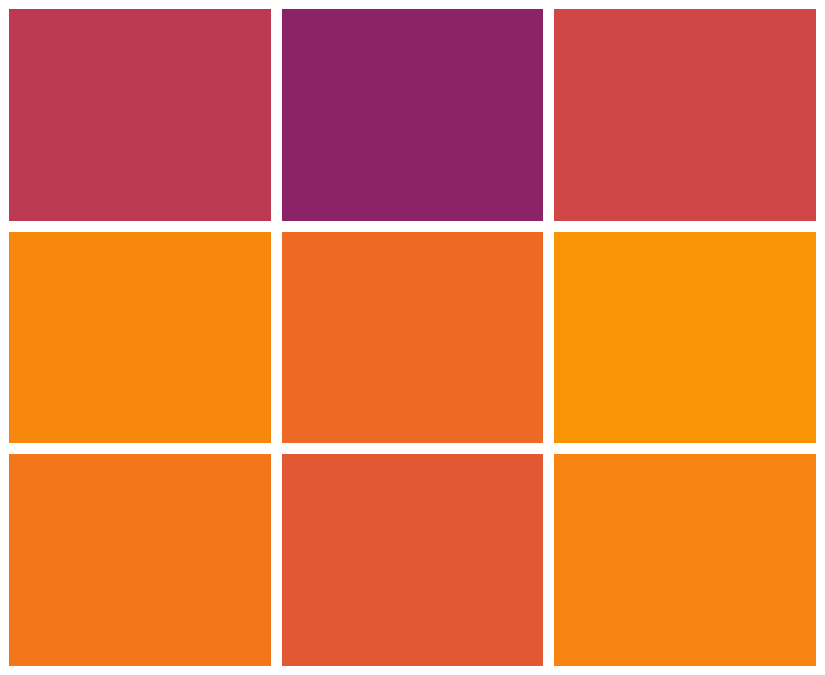}}}\hfil \\
                  
    \subfloat[Hedge]{{\includegraphics[width=0.15\textwidth]{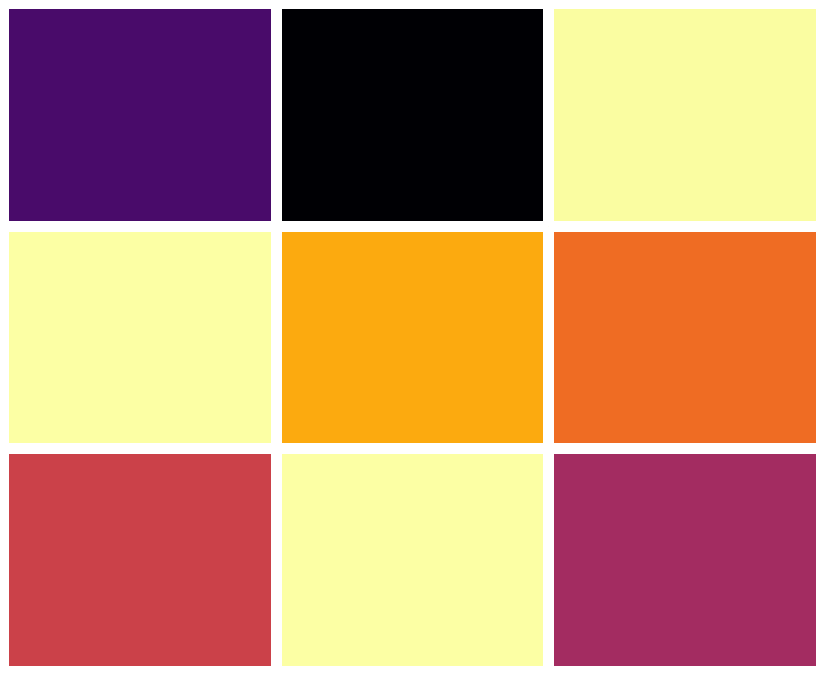}}
                  {\includegraphics[width=0.15\textwidth]{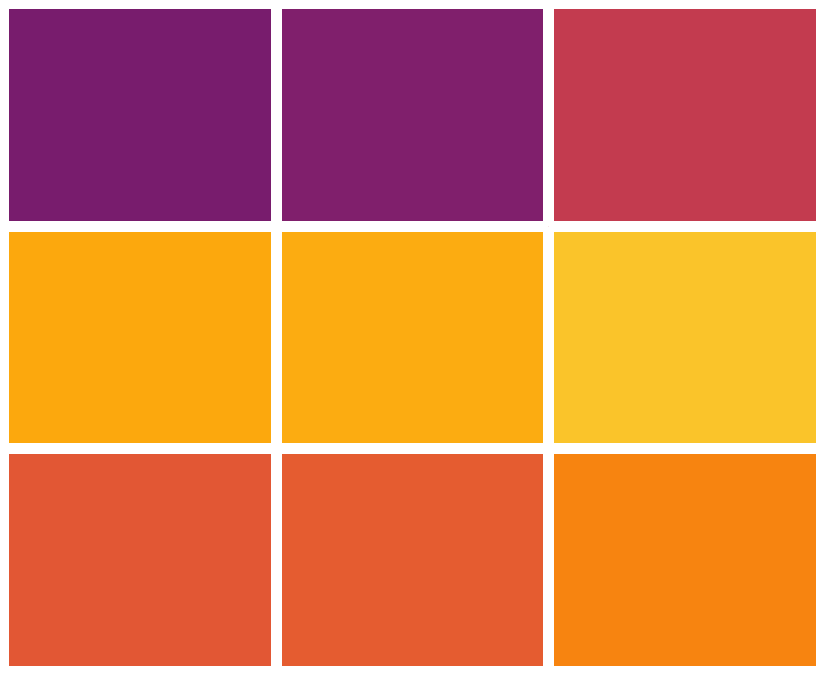}}} \hfil
    \subfloat[Hedge+]{{\includegraphics[width=0.15\textwidth]{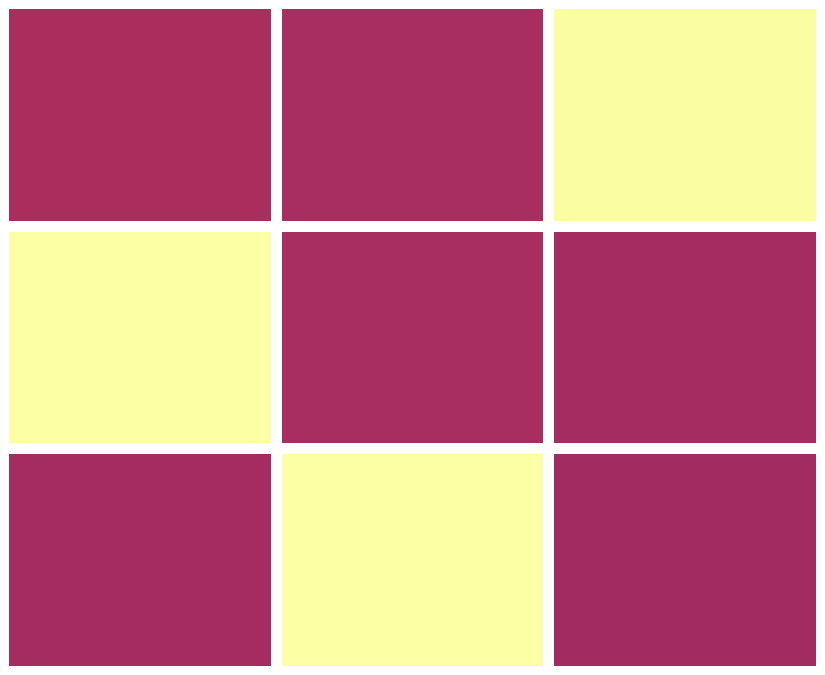}}
                  {\includegraphics[width=0.15\textwidth]{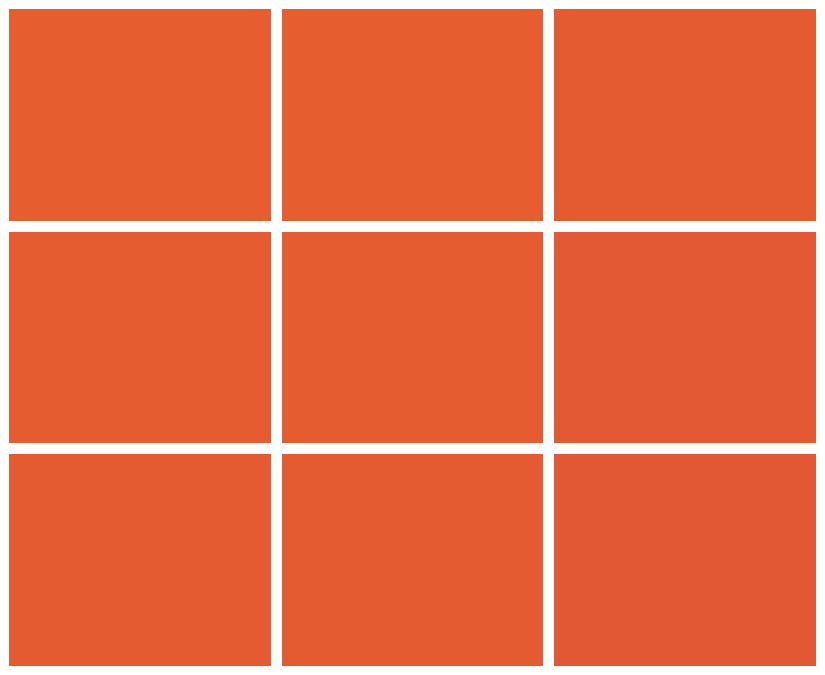}}} \hfil
    \subfloat[LCFR]{{\includegraphics[width=0.15\textwidth]{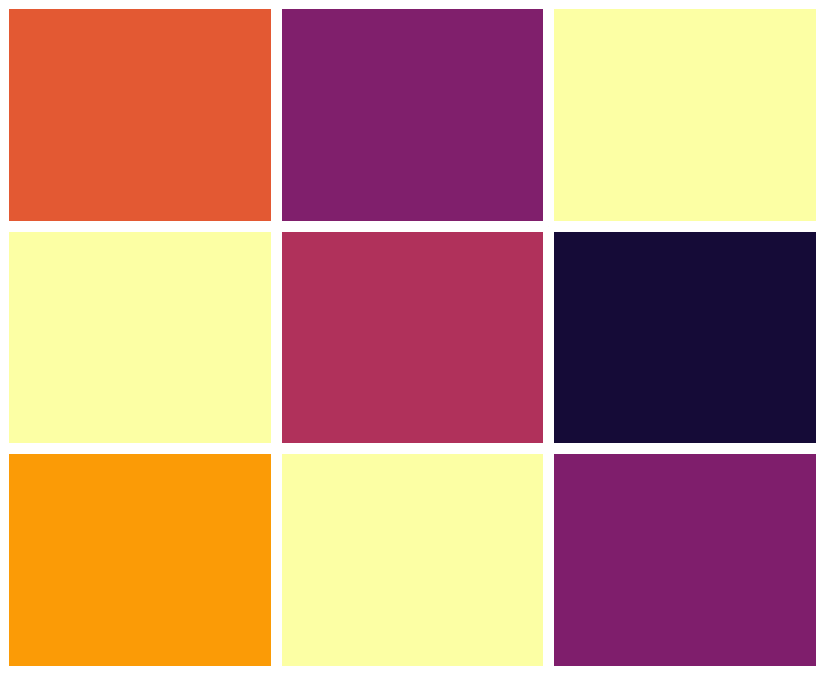}}
                  {\includegraphics[width=0.15\textwidth]{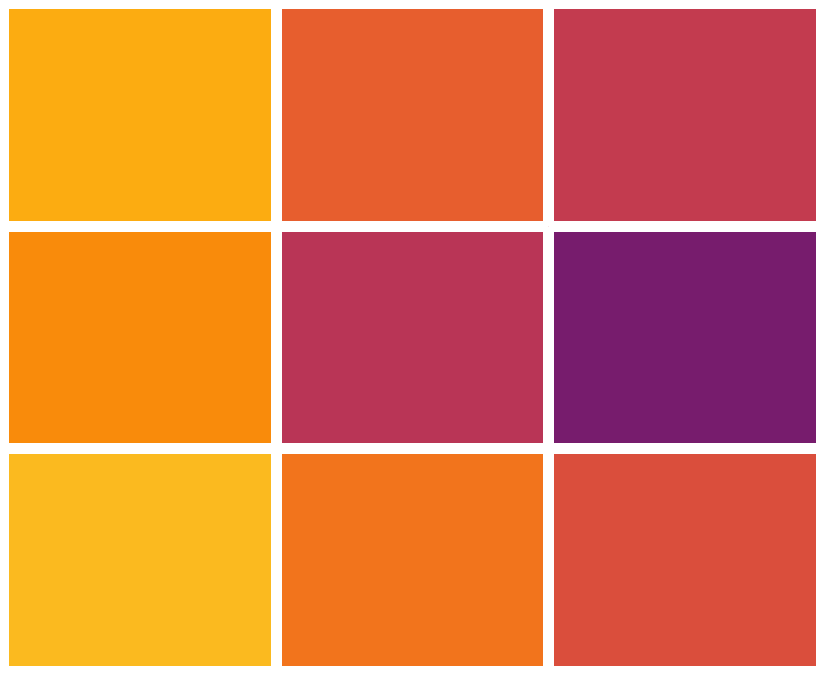}}} \hfil \\

    \subfloat[PCFR]{{\includegraphics[width=0.15\textwidth]{figures/heatmaps/PCFR.pdf}}
                  {\includegraphics[width=0.15\textwidth]{figures/heatmaps/PCFR-decorrelated.pdf}}} \hfil
    \subfloat[PCFR+]{{\includegraphics[width=0.15\textwidth]{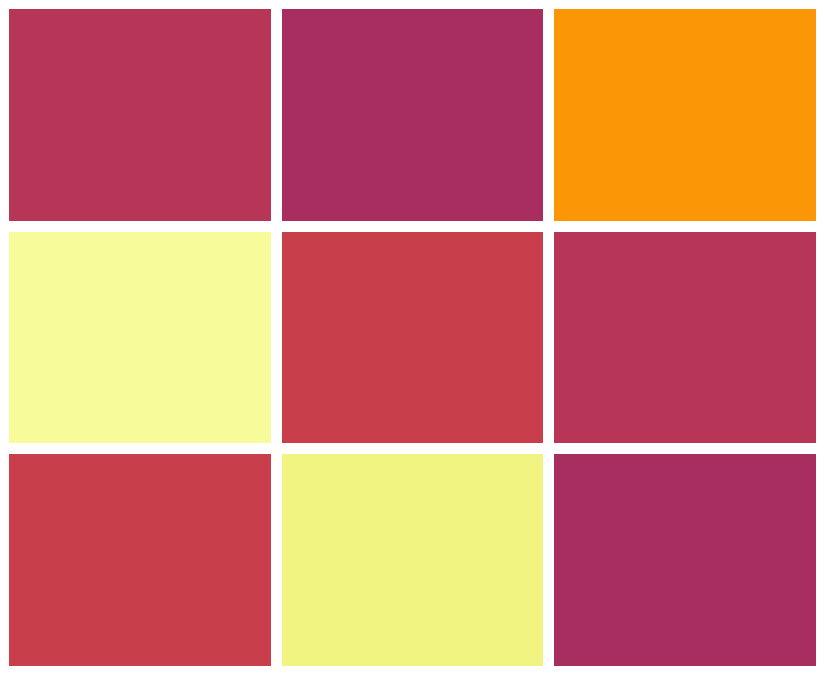}}
                  {\includegraphics[width=0.15\textwidth]{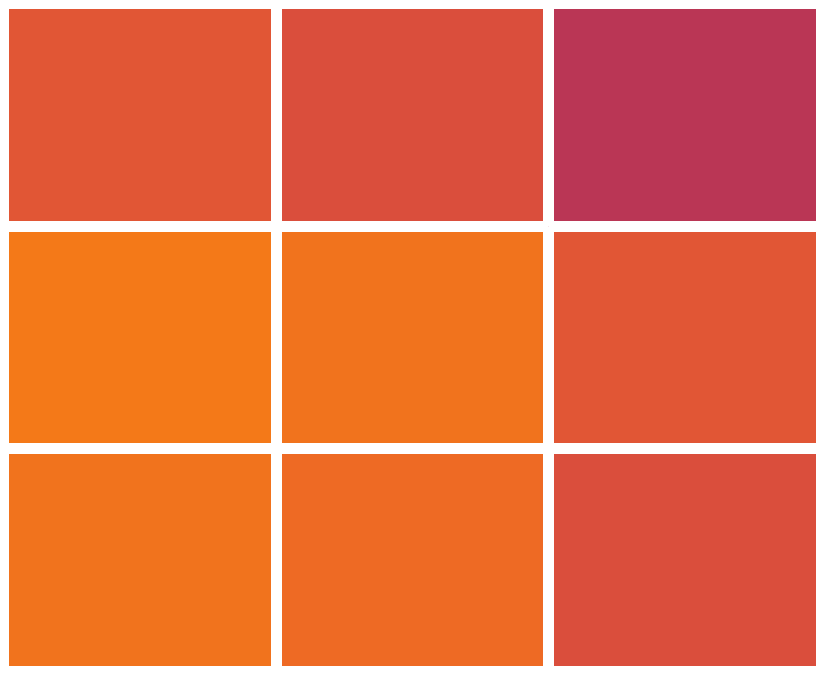}}} \hfil
    \subfloat[SPCFR]{{\includegraphics[width=0.15\textwidth]{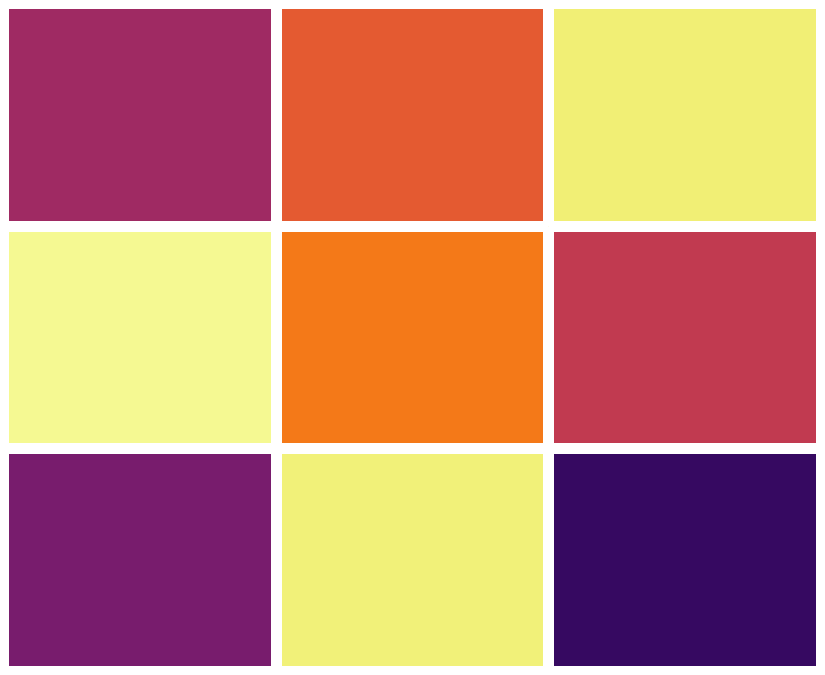}}
                  {\includegraphics[width=0.15\textwidth]{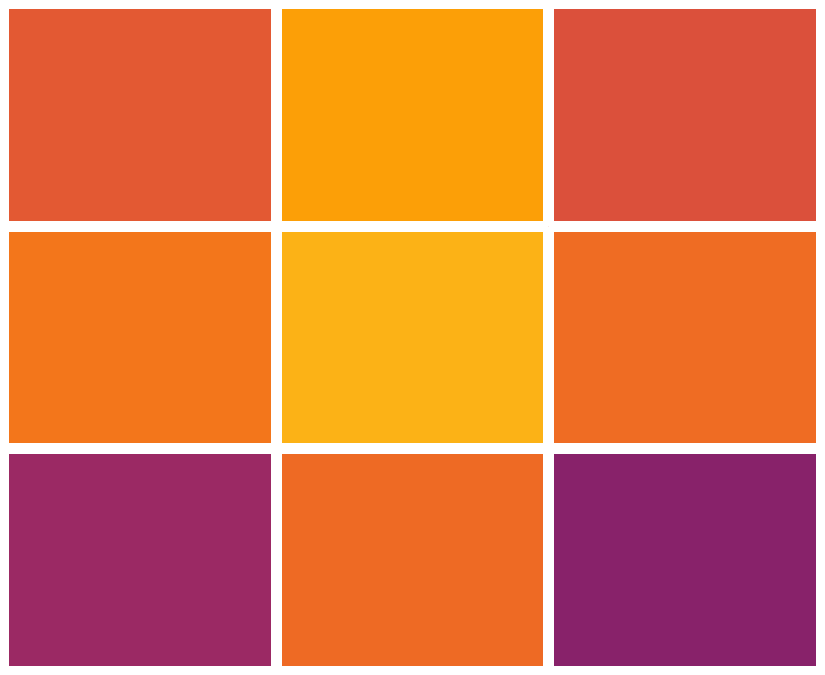}}} \hfil\\

    \subfloat[SPCFR+]{{\includegraphics[width=0.15\textwidth]{figures/heatmaps/SPCFR+.pdf}}
                  {\includegraphics[width=0.15\textwidth]{figures/heatmaps/SPCFR+-decorrelated.pdf}}} \hfil
    \subfloat[NPCFR]{{\includegraphics[width=0.15\textwidth]{figures/heatmaps/NPCFR.pdf}}
                  {\includegraphics[width=0.15\textwidth]{figures/heatmaps/NPCFR-decorrelated.pdf}}} \hfil
    \subfloat[NPCFR+]{{\includegraphics[width=0.15\textwidth]{figures/heatmaps/NPCFR+.pdf}}
                  {\includegraphics[width=0.15\textwidth]{figures/heatmaps/NPCFR+-decorrelated.pdf}}} \hfil

    \includegraphics[width=0.7\textwidth]{figures/heatmaps/colorbar.pdf}

    \caption{The empirical average joint strategy profiles found by regret minimizers $\overline{\jointstrategy}^T$ (left) and its marginalized version (right) found on {\tt biased\_shapley}$(0,1/2)$ after $T=2^{14}$ steps; see Eq.~\eqref{eq: shapley cce}. 
    Darker colors indicate higher probability under $\overline{\jointstrategy}^T$, and minimal differences between left and right figures imply the joint strategy is marginalizable.}
    \label{fig: app: heatmaps}
\end{figure*}

\newcolumntype{P}[1]{>{\centering\arraybackslash}p{#1}}

\begin{table*}[t]
\centering
\begin{tabular}{|c||P{5ex}|P{5ex}|P{5ex}|P{5ex}|P{6.5ex}|P{6.5ex}|P{5ex}|P{5ex}|P{5ex}|P{5ex}||P{5ex}|P{5ex}|}
\hline
NashGap & \multicolumn{2}{c|}{CFR$^{(+)}$} & \multicolumn{2}{c|}{PCFR$^{(+)}$} & DCFR & LCFR & \multicolumn{2}{c|}{SPCFR$^{(+)}$} & \multicolumn{2}{c||}{Hedge$^{(+)}$} & \multicolumn{2}{c|}{NPCFR$^{(+)}$} \\ \hline\hline
$10^{-1}$& {\bf 1} & 0.96 & {\bf 1} & 0.97 &  0.97 & 0.97 & {\bf 1} & {\bf 1} & {\bf 1} & {\bf 1} & {\bf 1}  &  {\bf 1} \\ \hline
$10^{-2}$& 0.78 & 0.09 & {\bf 1} & 0.09 & 0.09 & 0.42 & {\bf 1} & 0.09 & {\bf 1} & 0.36 & {\bf 1} & {\bf 1} \\ \hline
$10^{-3}$& 0.09 & 0.02 & 0.91 & 0.02 & 0.02 & 0.02 & {\bf 1} & 0.02 & {\bf 1} & 0.06 & {\bf 1} &  {\bf 1} \\ \hline
$10^{-4}$& 0.02 & 0 & 0.52 & 0 & 0 & 0   & 0.54 & 0 & 0.53 & 0 & {\bf 1} &  {\bf 1} \\ \hline
$10^{-5}$& 0 & 0 & 0.02  & 0 & 0 & 0 & 0.11 & 0 & 0.25 & 0 & 0.14 & {\bf 1}  \\ \hline
\end{tabular}
\caption{The fraction of games from {\tt biased\_shapley} each algorithm can solve to a given NashGap within $2^{14}=16,384$ steps. 
For the algorithms marked $^{(+)}$, the left column show the standard version, while the right shows the `plus'.
}
\label{tab: shapley threasholds extended}
\end{table*}

\begin{table*}[]
\centering
\begin{tabular}{l|cc|cc|cc|}
\cline{2-7}
                                 & \multicolumn{2}{c|}{{\tt biased\_shapley}}   & \multicolumn{2}{c|}{{\tt leduc\_poker}}      & \multicolumn{2}{c|}{{\tt three\_player\_leduc}} \\ \cline{2-7} 
                                 & \multicolumn{1}{c|}{NPCFR$^{(+)}$} & Other & \multicolumn{1}{c|}{NPCFR$^{(+)}$} & Other & \multicolumn{1}{c|}{NPCFR$^{(+)}$}     & Other     \\ \hline
\multicolumn{1}{|l|}{Training}   & \multicolumn{1}{c|}{0.2}         & -     & \multicolumn{1}{c|}{5}         & -     & \multicolumn{1}{c|}{42}             & -         \\ \hline
\multicolumn{1}{|l|}{Evaluation} & \multicolumn{1}{c|}{0.1}         &   0.01    & \multicolumn{1}{c|}{2}         &  2     & \multicolumn{1}{c|}{2}             &  2         \\ \hline
\end{tabular}
\caption{The number of GB of memory required to train and deploy the meta-learned regret minimizer, compared to the prior regret minimizers. 
We used $T=32$ meta-training steps to obtain the training values.
For evaluation, the memory requirements don't depend on the number of steps.}
\label{app: tab: memory requirements}
\end{table*}

Convergence to a Nash equilibrium using our proposed method requires minimizing both the distance of the joint strategy to a CCE and our meta-loss. A marginalizable joint strategy that is not a CCE is also not a Nash equilibrium, so verifying that the non-meta-learned algorithms converge to a CCE validates our meta-learning approach.
Figure~\ref{fig: app: biased shapley} shows that all evaluated non-meta-learned regret minimizers converge to a CCE in {\tt biased\_shapley}---
implying their failure to converge to a Nash equilibrium is a result of correlation between the strategies of the individual players.

Figure~\ref{fig: app: heatmaps} shows the joint strategy profiles found by the remaining regret minimizers in {\tt biased\_shapley}.
The evaluated non-meta-learned algorithms all converge to an approximation of
$\jointstrategy^*$~(\ref{app: biased shapley}) that is not marginalizable.
These heatmaps offer a visual interpretation of the correlation in the joint strategy that is minimized through our meta-loss.
NPCFR and NPCFR$^+$ find marginalizable joint strategies that are necessarily closer approximations to a Nash equilibrium.

Finally, we study the chance that, on a random sample from {\tt biased\_shapley}$(0,1/2)$, a particular regret minimizer finds a strategy with a given NashGap.
To do this, we sample 64 games and run each algorithm for $2^{14}$ regret minimization steps.
We evaluate the strategy every $2^{k}$ steps, $k\in{0, \dots, 14}$, and record the lower NashGap encountered along the trajectory.
We summarize our findings in Table~\ref{tab: shapley threasholds extended}.
The non-meta-learned algorithms exhibit poor performance, typically being unable to obtain even ${\rm NashGap} \approx 10^{-3}$.
This is particularly true for the `plus' versions.
In contrast, the meta-learned algorithms can reliably obtain even very precise approximations of a Nash equilibrium.

\subsection{Two-Player Leduc Poker}
\label{app: extra res: 2leduc}

\begin{figure*}[t]
    \centering
    \includegraphics[width=0.48\textwidth]{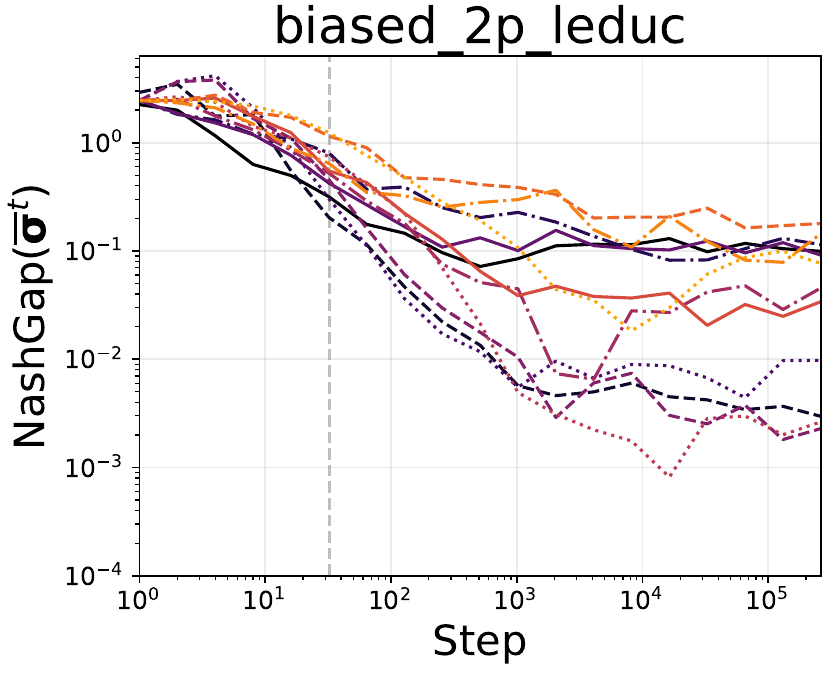}
    \includegraphics[width=0.48\textwidth]{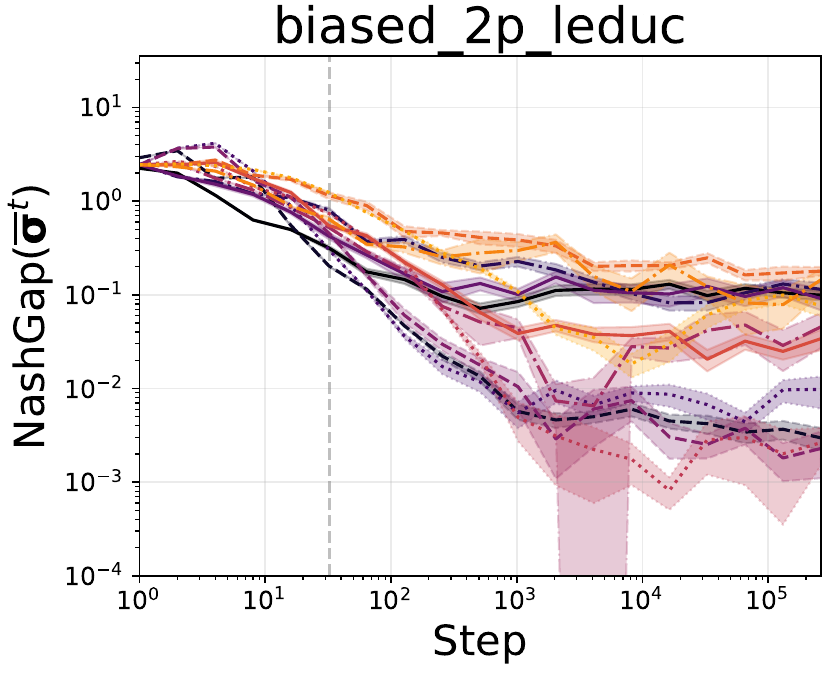}\\
    
    \includegraphics[width=0.8\textwidth]{figures/legend.pdf}
    \caption{
    Comparison of non-meta-learned algorithms (\CFR$^{(+)}$, \PCFR$^{(+)}$, \SPCFR$^{(+)}$, Hedge$^{(+)}$, DCFR, and LCFR) with meta-learned algorithms NPCFR$^{(+)}$. NPCFR$^{(+)}$ was trained on \twoPlLeduc{}.
    The figures show NashGap of the average strategy profile $\overline{\strategy}^t$. 
    Vertical dashed lines separate the training (up to $T=32$ steps) and the generalization (from $T$ to $2^{18}=262,144$) regimes.
    The colored areas show standard errors.
    }
    \label{fig: 2 pl leduc}
\end{figure*}

\begin{table*}[t]
\centering
\begin{tabular}{|c||P{5ex}|P{5ex}|P{5ex}|P{5ex}|P{6.5ex}|P{6.5ex}|P{5ex}|P{5ex}|P{5ex}|P{5ex}||P{5ex}|P{5ex}|}
\hline
NashGap & \multicolumn{2}{c|}{CFR$^{(+)}$} & \multicolumn{2}{c|}{PCFR$^{(+)}$} & DCFR & LCFR & \multicolumn{2}{c|}{SPCFR$^{(+)}$} & \multicolumn{2}{c||}{Hedge$^{(+)}$} & \multicolumn{2}{c|}{NPCFR$^{(+)}$} \\ \hline\hline
$10^{-1}$& {\bf 1} & {\bf 1} & {\bf 1} & {\bf 1} & {\bf 1} & {\bf 1} & {\bf 1} & {\bf 1} & 0 & 0.29 & {\bf 1} & {\bf 1} \\ \hline
$10^{-2}$& 0 & {\bf 1} & 0.03 & {\bf 1} &  0.13 & 0 & 0.54 & {\bf 1} & 0 & 0.29 & 0.84 & {\bf 1} \\ \hline
$10^{-3}$& 0 & 0 & 0 & 0.87 &  0 & 0 & 0 & 0.72 & 0 & 0 & 0.73 & {\bf 0.98} \\ \hline
$10^{-4}$& 0 & 0 & 0 & 0.48 &  0 & 0 & 0 & 0.21 & 0 & 0 & 0.73 & {\bf 0.96} \\ \hline
$10^{-5}$& 0 & 0 & 0  & 0.16 & 0 & 0 & 0 & 0.11 & 0 & 0 & 0.73 & {\bf 0.96}  \\ \hline
\end{tabular}
\caption{The fraction of games from \twoPlLeduc{} each algorithm can solve to a given NashGap within $2^{18}=262,144$ steps.  
For the algorithms marked $^{(+)}$, the left column show the standard version, while the right shows the `plus'.
}
\label{tab: 2leduc threasholds extended}
\end{table*}

Following the experiments in normal-form, we investigate NashGap of strategies produced by each regret minimizer in Figure~\ref{fig: 2 pl leduc}.
We find that the NashGap average across the games is similar between the non-meta-learned and meta-learned algorithms.
However, as is shown in Table~\ref{tab: 2leduc threasholds extended}, the meta-learned algorithms reliably encounter very good approximations of Nash equilibria along the regret minimization trajectory.
We speculate that this is a result of suboptimal generalization.
Our predictors $\pi(\cdot|\theta)$ were meta-learned to minimize the correlations only for $T=32$ steps, but evaluated on a much longer scale.
The strategies of NPCFR$^{(+)}$ thus remain uncorrelated only to a point, which is insufficient for the $\mathcal{O}(1/\sqrt{T})$ to decay sufficiently close to zero, see Eq.~\eqref{eq: marginal thm bound}.

For the non-meta-learned algorithms, the CCE gap~\eqref{eq: cfr bound} shows that at the end of evaluation, the strategy is close to a CCE. However, the NashGap remains high, implying that the CCE must be strictly correlated.
Note this doesn't guarantee the non-meta-learn algorithms will not eventually converge closer to a Nash equilibrium.

\subsection{Three-Player Leduc Poker}
\label{app: extra res: 3leduc}

% \begin{figure}[t]
%     \centering
%     \includegraphics[width=0.44\textwidth]{figures/3_pl_leduc/nash_conv_3leduc_poker.pdf}
%     % \includegraphics[width=0.32\textwidth]{figures/shapley/nash_conv_current_strategy_biased_shapley.pdf}
%     \includegraphics[width=0.44\textwidth]{figures/3_pl_leduc/cce_cfr_bound_3leduc_poker.pdf}\\
    
%     \includegraphics[width=0.8\textwidth]{figures/3_pl_leduc/legend.pdf}
%     \caption{
%     Comparison of non-meta-learned algorithms (\CFR$^{(+)}$, \PCFR$^{(+)}$, \SPCFR$^{(+)}$, Hedge$^{(+)}$, DCFR, and LCFR) with meta-learned algorithms NPRM$^{(+)}$.
%     We show algorithms meta-learned and evaluated on {\tt three\_player\_leduc}.
%     The figures show NashConv of the average joint strategy profile $\overline{\jointstrategy}^t$. 
%     Vertical dashed lines separate the training (up to $T=32$ steps) and the generalization (from $T$ to $2^{18}=262,144$ steps) regimes. 
%     }
%     \label{fig: 3leduc}
% \end{figure}

\begin{figure*}[t]
    \centering
    \includegraphics[width=0.48\textwidth]{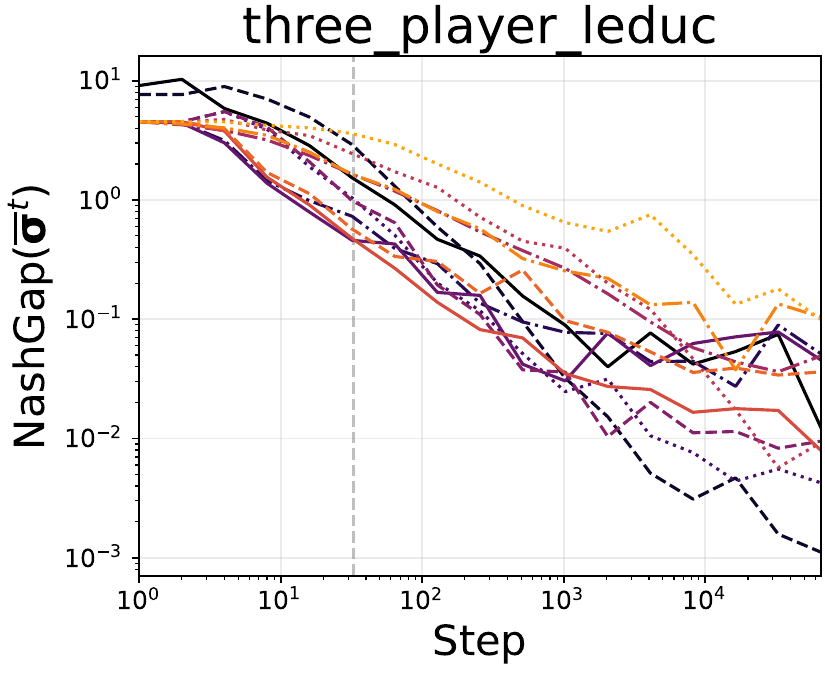}\\
    
    \includegraphics[width=0.8\textwidth]{figures/legend.pdf}
    \caption{
    Comparison of non-meta-learned algorithms (\CFR$^{(+)}$, \PCFR$^{(+)}$, \SPCFR$^{(+)}$, Hedge$^{(+)}$, DCFR, and LCFR) with meta-learned algorithms NPCFR$^{(+)}$. NPCFR$^{(+)}$ was trained on {\tt three\_player\_leduc}.
    The figures show NashGap of the average strategy profile $\overline{\strategy}^t$. 
    Vertical dashed lines separate the training (up to $T=32$ steps) and the generalization (from $T$ to $2^{16}=65,536$) regimes.
    }
    \label{fig: 3 pl leduc}
\end{figure*}
\begin{table*}[t]
\centering
\begin{tabular}{|P{6ex}|P{6ex}|P{6ex}|P{6ex}|P{6.5ex}|P{6.5ex}|P{6ex}|P{6ex}|P{6ex}|P{6ex}||P{6ex}|P{6ex}|}
\hline
 \multicolumn{2}{|c|}{CFR$^{(+)}$} & \multicolumn{2}{c|}{PCFR$^{(+)}$} & DCFR & LCFR & \multicolumn{2}{c|}{SPCFR$^{(+)}$} & \multicolumn{2}{c||}{Hedge$^{(+)}$} & \multicolumn{2}{c|}{NPCFR$^{(+)}$} \\ \hline\hline
 0.027 & 0.004 & 0.030 & 0.008 & 0.008 & 0.034 & 0.036 & 0.006 & 0.038 & 0.099 & 0.012 & {\bf 0.001} \\ \hline
\end{tabular}
\caption{The lowest value of NashGap encountered on {\tt three\_player\_leduc} for each algorithm.
For the algorithms marked $^{(+)}$, the left column show the standard version, and right shows the `plus'.
}
\label{tab: 3leduc poker best nashconv}
\end{table*}

Finally, we investigate a fixed game, in particular the three-player version of Leduc poker; see Appendix~\ref{app: games: leduc poker}. 
The NashGap of strategies produced by each regret minimizer in Figure~\ref{fig: 3 pl leduc}.
Like in the case of \twoPlLeduc{}, we evaluate the strategy of each algorithm only at $2^k$ steps.
We evaluate algorithms on {\tt three\_player\_leduc} for $2^{16} = 65,536$ steps in total.
We report the lowest value of NashGap each algorithm was able to find in Table~\ref{tab: 3leduc poker best nashconv}.
Among the algorithms we evaluated, those using alternating updates (for us that mean algorithms marked `plus', and DCFR) consistently outperform their standard counterparts.
One notable exception is the Hedge algorithm, for which the `plus' modification performed considerably worse.
The best approximation of a Nash equilibrium among the non-meta-learned algorithms was obtained by CFR$^+$, yielding ${\rm NashGap}=0.004$.

The meta-learned versions of NPCFR, resp. NPCFR$^{(+)}$ outperform all algorithms without / with the `plus' modification.
The lowest value of NashGap we found was obtained by NPCFR$^+$ with ${\rm NashGap}=0.001$.
For this particular experiment, we were inspired by the strong empirical performance of CFR$^+$, and we modified the form of the prediction to
\begin{equation*}
    \prediction^{t+1} \gets \alpha \pi(\regret^t, \cumregret^t, \vv{e}_s|\theta), 
\end{equation*}
see also Algorithm~\ref{algo:nprm}, line 9.

\subsection{Out-of-Distribution Convergence}
\label{app: extra res: out of dist}

The meta-learned algorithms are implicitly tailored to the domain they were trained on.
As such, when evaluated out-of-distribution, they can lose their empirical performance.
Figure~\ref{fig: out of dictibution} illustrates this by showing evaluation of NPCFR$^{(+)}$ on {\tt biased\_shapley}$(-1,0)$.
The algorithms were meta-trained on {\tt biased\_shapley}$(0,1/2)$, see Section~\ref{ssec: matrix games} and Figure~\ref{fig: app: biased shapley}.

Since NPCFR$^{(+)}$ are regret minimizers, they are guaranteed to converge to a CCE.
However, they both fail to converge to a Nash equilibrium out-of-distribution.

\begin{figure*}[t]
    \centering
    \includegraphics[width=0.48\textwidth]{figures/shapley/nash_conv_biased_shapley_error.pdf}
    \includegraphics[width=0.48\textwidth]{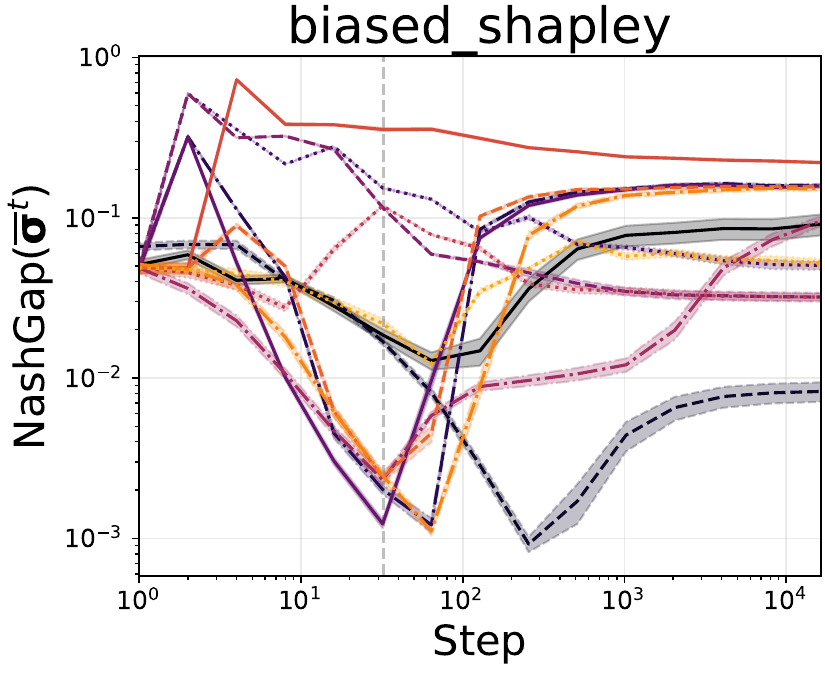}\\
    
    \includegraphics[width=0.8\textwidth]{figures/legend.pdf}
    \caption{
    Comparison of non-meta-learned algorithms (\CFR$^{(+)}$, \PCFR$^{(+)}$, \SPCFR$^{(+)}$, Hedge$^{(+)}$, DCFR, and LCFR) with meta-learned algorithms NPCFR$^{(+)}$. NPCFR$^{(+)}$ was trained on {\tt biased\_shapley}$(0,1/2)$, and evaluated in-distribution (left) or out-of-distribution on {\tt biased\_shapley}$(-1,0)$ (right).
    The figures show NashGap of the average strategy profile $\overline{\strategy}^t$, the right figure includes standard errors. 
    Vertical dashed lines separate the training (up to $T=32$ steps) and the generalization (from $T$ to $2^{14}=16,384$ steps) regimes. 
    The colored areas show standard errors.
    }
    \label{fig: out of dictibution}
\end{figure*}

%%%%%%%%%%%%%%%%%%%%%%%%%%%%%%%%%%%%%%%%%%%%%%%%%%%%%%%%%%%%

% \clearpage
% \input{sections/neurips25_checklist}

\end{document}